\documentclass[letterpaper, 10 pt, journal]{ieeetran}

\IEEEoverridecommandlockouts                              

\usepackage{graphics} 
\usepackage{amsmath} 
\usepackage{amssymb} 
\usepackage{amsthm}
\usepackage{subfigure}
\usepackage{epstopdf}
\usepackage{mathtools}
\usepackage{enumitem}
\usepackage{color}
\usepackage{xfrac}
\usepackage{lipsum}
\usepackage{cuted}
\usepackage{array}
\usepackage{cases}
\usepackage[numbers,sort&compress]{natbib}

\newtheorem{thm}{Theorem}
\newtheorem{defn}{Definition}
\newtheorem{cor}{Corollary}
\newtheorem{lem}{Lemma}
\newtheorem{remark}{Remark}
\newtheorem{assumption}{Assumption}
\newtheorem{example}{Example}

\title{Synchronization and Balancing around Simple Closed Polar Curves with Bounded Trajectories and Control Saturation}

\author{Aditya Hegde,~\IEEEmembership{Graduate~Student~Member,~IEEE}\thanks{A. Hegde is with 
		Department of Aerospace Engineering, Indian Institute of Science, 
		Bangalore 560012, India (e-mail: adityahegde@iisc.ac.in).}, and Anoop Jain,~\IEEEmembership{Member,~IEEE}\thanks{A. Jain is with Department of Electrical Engineering, 
		Indian Institute of Technology, Jodhpur 342037, India (e-mail: 
		anoopj@iitj.ac.in).}}
	
\begin{document}
\maketitle
		

\begin{abstract}
The problem of synchronization and balancing around simple closed polar curves is addressed for unicycle-type multi-agent systems. Leveraging the concept of barrier Lyapunov function in conjunction with bounded Lyapunov-like curve-phase potential functions, we propose distributed feedback control laws and show that the agents asymptotically stabilize to the desired closed curve, their trajectories remain bounded within a compact set, and their turn-rates adhere to the saturation limits. We also characterize the explicit nature of the boundary of this \emph{trajectory-constraining} set based on the magnitude of the safe distance of the exterior boundary from the desired curve. We further establish a connection between the perimeters and areas of the trajectory-constraining set with that of the desired curve. We obtain bounds on different quantities of interest in the post-design analysis and provide simulation results to illustrate the theoretical findings.   
\end{abstract}



\section{Introduction}\label{introduction}

Constraints are an integral part of any practical system. Depending upon the physical and operational requirements, systems may have several constraints, and it is a challenging task to design stabilizing controllers for their safe operation. Particularly, in the context of multi-agent systems, one such problem is to design distributed controllers such that the agents (or vehicles) do not transgress the given workspace while stabilizing to a desired collective formation. Such problems find numerous applications pertaining to surveillance and patrolling across territories where it is desired that multiple vehicles maintain a prescribed formation and their trajectories do not cross the border due to safety considerations. Besides, there are several other applications like autonomous driving, space missions, ocean explorations, etc., where it is required that the vehicles' trajectories remain bounded within a certain region of interest. By constraining agents' trajectories, not only is their safe operation assured, but also, their interaction topology, which usually relies on sensors with limited sensing range, can be preserved. 

In recent years, the distributed control algorithms in this direction are derived using the concepts of Lyapunov-like barrier functions, meeting system constraints like collision avoidance, proximity maintenance, and bounded control input \cite{panagou2013multi, panagou2015distributed, glotfelter2017nonsmooth, han2019robust, verginis2019closed, lee2011constrained}. The control design methodologies in these works rely on the composition of the well-known notions of Control Lyapunov Function (CLF) and Control Barrier Function (CBF) \cite{romdlony2014uniting, romdlony2016stabilization, ames2016control, ames2019control}. The Barrier Lyapunov Function (BLF) is a special class of CLF, which grows to infinity when its argument approaches the desired limits. Several variants of BLF have been used in literature to solve different problems associated with single and multi-agent systems; for instance, recentered BLF \cite{panagou2013multi, panagou2015distributed}, parametric BLF \cite{han2019robust}, logarithmic BLF \cite{tee2009barrier, tee2011control}, tangent-type BLF \cite{tang2013tangent}, integral-BLF \cite{he2014top}, etc. In this paper, a collection of logarithmic BLFs \cite{tee2009barrier, tee2011control} associated to each agent is combined with bounded Lyapunov-like curve-phase potentials to derive the stabilizing feedback control laws. The proposed curve-phase patterns are characterized by relative arc-lengths of the agents' motion along the desired curve. This concept finds numerous applications in the domain of sensor networks where the signals received at the receiving end may be weak due to the presence of a dense medium. Therein, synchronized and balanced phase arrangements of the agent/sensor networks, operating at different levels in a circular formation, help in collecting information optimally \cite{leonard2007collective}.

However, in several applications, neither the territorial boundaries nor the regions, required to be tracked by the agents, are necessarily circular. For instance, environmental boundaries, defined by the level sets of a scalar physical field, like the concentration of oil spills or the intensity of a light source, are essentially non-circular and can be better estimated by non-convex curves \cite{ovchinnikov2015cooperative,brinon2015distributed,brinon2019multirobot}. Motivated by these aspects, we stabilize in this work the motion of agents around simple closed polar curves in curve-phase synchronization or balancing. This generalizes the results in \cite{jain2019trajectory} where all-to-all communication-based control laws were used for stabilizing the collective motion around circular orbits$-$a specific case of simple closed curves. Contrary to \cite{jain2019trajectory}, in this paper, we introduce a parametric-phase model to account for \emph{convexity} of the desired curve, define a generalized notion of synchronization and balancing, propose distributed control laws, and rigorously analyze the nature of the trajectory-constraining region, and its perimeter and area. 
Throughout the paper, we mean by the term \emph{simple closed polar curves} $-$ simple closed curves expressed in polar form. 

The majority of \emph{prior} research in this direction relies on quadratic Lyapunov functions and does not impose any requirement on agents' trajectories. Whilst the notions of synchronization and balancing are discussed in \cite{paley2008stabilization}, the results are limited to \emph{skewed superellipses}, a special class of convex curves, and no restrictions are imposed on the agents' trajectories and the control input. Moreover, \cite{chen2015formation,sabattini2013closed,zhang2007coordinated} do not talk about the notions of synchronized and balanced curve-phases. In \cite{xu2017realizing}, simultaneous lane-keeping and speed regulation of the robots were realized using a quadratic programming framework. Unlike \cite{jain2019trajectory,paley2008stabilization,chen2015formation,sabattini2013closed,zhang2007coordinated,xu2017realizing}, in this work, we not only stabilize the agents around a general class of simple closed curves, comprising both convex and non-convex curves, but also achieve synchronized and balanced curve-phase patterns in their collective motion. We also assure that the agents' trajectories remain bounded during stabilization and the control input obeys the saturation limits, meeting the turn-rate constraints of a vehicle. 

\emph{Contributions:} By combining the idea of logarithmic BLF, along with, bounded Lyapunov-like curve-phase potentials, we derive feedback control laws to stabilize agents' motion around simple closed polar curves in synchronized and balanced curve-phases with bounded trajectories. The proposed controllers obey \emph{pre-specified} saturation limits and consider limited communication topology among the agents. To account for \emph{convexity} of the desired curve $\mathcal{C}$, we first propose a parametric-phase model to decide the evolution of the tracking point on $\mathcal{C}$ corresponding to an agent's heading. We show that the proposed controllers ensure that the agents' trajectories remain bounded within a compact set $\mathcal{B}_{\delta}$, characterized by the magnitude of the safe distance $\delta$ from the desired curve $\mathcal{C}$ and the unit normal vectors $\hat{g}_n$ to $\mathcal{C}$, and their turn-rates adhere to the saturation limits. We also characterize the explicit nature of the boundary of the set $\mathcal{B}_{\delta}$, under an assumption on $\delta$ and $\hat{g}_n$, motivated by practical applications. We show that there may exist multiple boundaries of the set $\mathcal{B}_{\delta}$, depending on $\delta$, and are constructed using the locus of the farthest points from the desired curve. We further establish a connection between the perimeters and areas of $\mathcal{B}_{\delta}$ with that of the desired curve $\mathcal{C}$. 
We further obtain analytical bounds on various signals in the post-design analysis and illustrate the results through a simulation example. 

\emph{Paper Structure:} Section~\ref{section2} describes notations, introduces the system model, and reviews some preliminary results. The idea of curvature control and the notion of synchronized and balanced curve-phases are discussed in Section~\ref{section3}. Section~\ref{section4} proposes control laws based on composite Lyapunov functions, and also describes the boundary, perimeter, and area of the trajectory-constraining set. Section~\ref{section5} obtains bounds on different signals of interest in both curve-phase synchronization and balancing. Section~\ref{section6} presents simulation results, before we conclude and present the future directions of work in Section~\ref{section7}.

\section{System Description and Some Background Results}\label{section2}
This section describes notations, introduces the system model, and reviews some basic results about BLF.

\subsection{Preliminaries}
The set of real, complex, natural, and positive (non-negative) real numbers is $\mathbb{R}$, $\mathbb{C}$, $\mathbb{N}$, and $\mathbb{R}_{>0} (\mathbb{R}_{\geq 0})$, respectively. The imaginary unit is $i = \sqrt{-1}$. The unit circle in the complex plane is the set $\mathbb{S}^1 \subset \mathbb{C}$. The $N$-torus is the set $\mathbb{T}^N = \mathbb{S}^1 \times \ldots \times \mathbb{S}^1$ ($N$ times), where,  $\times$ is the Cartesian product operator. The inner product of two complex numbers $z_1,z_2 \in \mathbb{C}$ is given by $\langle z_1, z_2 \rangle  = \Re(\bar{z}_1z_2)$, where $\bar{z}_1 \in \mathbb{C}$ is the complex conjugate of $z_1$. For vectors, we use the analogous boldface notation  $\langle \pmb{w}, \pmb{z} \rangle  = \Re(\pmb{w}^\ast\pmb{z})$ for $\pmb{w}, \pmb{z} \in \mathbb{C}^N$, where $\pmb{w}^\ast$ is the conjugate transpose of $\pmb{w}$. For $\pmb{\varphi} = [\varphi_1, \ldots, \varphi_N]^T \in \mathbb{T}^N$, the $N$ vector ${\rm e}^{i\pmb{\varphi}}$ is used to denote ${\rm e}^{i\pmb{\varphi}} = [{\rm e}^{i\varphi_1}, \ldots, {\rm e}^{i\varphi_N}]^T$. A differentiable map $f: \mathcal{D} \to\mathbb{R}$, $\mathcal{D} \subseteq \mathbb{R}^N$, has a gradient $\nabla_{\pmb{x}} f = \left[{\partial f}/{\partial x_1}, \ldots, {\partial f}/{\partial x_N}\right]^T$. We denote by $\pmb{0}_N = [0, \ldots, 0]^T \in \mathbb{R}^N$ and $\pmb{1}_N = [1, \ldots, 1]^T \in \mathbb{R}^N$. We often suppress arguments if clear from the context. 

A graph is a pair $\mathcal{G} = (\mathcal{V}, \mathcal{E})$, consisting of a finite set of vertices $\mathcal{V}$, and a finite set of edges $\mathcal{E} \subseteq \mathcal{V}\times \mathcal{V}$. The incidence matrix $\mathcal{M} \in \mathbb{R}^{|\mathcal{V}|\times|\mathcal{E}|}$ of graph $\mathcal{G}$ with an arbitrary orientation is defined such that, for each edge $e = (j, k) \in \mathcal{E}$, $[\mathcal{M}]_{je} = +1, [\mathcal{M}]_{ke} = -1$, and $[\mathcal{M}]_{\ell e} = 0$ for $\ell \neq j, k$. The Laplacian  $\mathcal{L} \in \mathbb{R}^{|\mathcal{V}|\times|\mathcal{V}|}$ of graph $\mathcal{G}$ is defined such that $[\mathcal{L}]_{jk} = |\mathcal{N}_j|$ if $j = k$, $[\mathcal{L}]_{jk} = -1$ if $k \in \mathcal{N}_j$, and $[\mathcal{L}]_{jk} = 0$ otherwise, where $|\mathcal{N}_j|$ is the cardinality of the set $\mathcal{N}_j$. The Laplacian quadratic form associated with graph $\mathcal{G}$ with $N$ nodes is defined as $Q_{\mathcal{L}}(\pmb{z}) = \langle \pmb{z}, \mathcal{L}\pmb{z} \rangle$ for $\pmb{z} \in \mathbb{C}^N$, which is positive semi-definite and is zero if and only if $\pmb{z} = z_0\pmb{1}_N$ for some $z_0 \in \mathbb{C}$. A graph $\mathcal{G}$ is circulant if and only if its Laplacian $\mathcal{L}$ is a circulant matrix \cite{davis2013circulant}. 

\begin{lem}[\hspace{-.1pt}\cite{davis2013circulant}]\label{lem_circulant_matrices}
	Let $\mathcal{L}$ be the Laplacian of an undirected circulant graph $\mathcal{G}$ with $N$ vertices. Define $\chi_k \coloneqq (k-1)2\pi/N$, for $k = 1, \ldots, N$. Then, the vectors $\pmb{f}^{(\ell)} \coloneqq {\rm e}^{i(\ell-1)\pmb{\chi}},~~\ell = 1, \ldots, N$, form a basis of $N$ orthogonal eigenvectors of $\mathcal{L}$. The unitary matrix $\mathcal{F}$, whose columns are the $N$ (normalized) eigenvectors $(1/\sqrt{N})\pmb{f}^{(\ell)}$, diagonalizes $\mathcal{L}$, that is, $\mathcal{L} = \mathcal{F}\Lambda \mathcal{F}^\ast$, where $\Lambda \coloneqq \text{diag}\{0, \lambda_2, \ldots,  \lambda_N\} \succeq 0$ is the (real) diagonal matrix of the eigenvalues of $\mathcal{L}$, and $\mathcal{F}^\ast$ denotes the conjugate transpose of $\mathcal{F}$.
\end{lem}

The following definitions about parametric curves are stated from \cite{tapp2016differential,pressley2010elementary}. Let $\alpha: [a, b] \to \mathbb{R}^2, t \mapsto \alpha(t)$ be a planar differentiable curve, parameterized by $t$. The curve $\alpha$ is said to be regular if $d\alpha/d t = \dot{\alpha}(t) \neq 0$ for all $t \in [a, b]$. A closed plane curve is a regular parameterized curve $\alpha$ such that $\alpha(a) = \alpha(b)$ and all the derivatives agree at $a$ and $b$; that is, $\dot{\alpha}(a) = \dot{\alpha}(b), \ddot{\alpha}(a) = \ddot{\alpha}(b)$, and so on. The curve $\alpha$ is \emph{simple} if it has no further self-intersections; that is, if $t_1, t_2 \in [a, b), t_1 \neq t_2$, then $\alpha(t_1) \neq \alpha(t_2)$. Further, $\alpha$ is periodic if there is a number $T > 0$ such that $\alpha(t + T) = \alpha(t)$ for all $t$, and the smallest such number $T$ is called the period of $\alpha$. It is clear that the simple closed curve $\alpha$ is a periodic curve with period $T = b-a$. A closed curve is said to be convex if the region it encloses is a convex set, else it is called non-convex. A \emph{re-parametrization} of $\alpha$ is a function of the form $\breve{\alpha} = \alpha \circ \varrho: [\breve{a}, \breve{b}] \to \mathbb{R}^2$, where $\varrho: [\breve{a}, \breve{b}] \to [a, b]$ is a smooth bijective map with nowhere-vanishing derivative, that is, $\dot{\varrho}(t) \neq 0$ for all $t \in [\breve{a}, \breve{b}]$. According to the \emph{Jordan Curve Theorem} [\cite{tapp2016differential}, pg. 62], any simple closed curve $\alpha$ in the plane has an `interior' (denote by $\text{int}(\alpha)$) and an `exterior' (denote by $\text{ext}(\alpha)$).

\subsection{System model}
A group of $N$ identical agents, moving in the $\mathbb{R}^2$ plane, is considered. For simplicity, a map  $(p, q) \mapsto p + iq$ is used to transform the $\mathbb{R}^2$ plane to the $\mathbb{C}$ plane. The position and heading of the $k^\text{th}$ agent are $r_k = x_k + iy_k \in \mathbb{C}$ and $\theta_k \in \mathbb{S}^1$. We assume that the agents move with unit speed and their velocity vectors can be expressed as $\dot{r}_k = {\rm e}^{i\theta_k} = \cos\theta_k + i\sin\theta_k \in \mathbb{C}, \forall k$. The consideration of constant speed is motivated by several practical applications pertaining to unmanned aerial vehicles \cite{jain2019trajectory}. With these notations, the motion of the agents is represented by
\begin{equation}\label{modelNew}
\dot{r}_k  = {\rm e}^{i\theta_k};~~~\dot{\theta}_k = u_k,~~~ k = 1, \ldots, N,
\end{equation}
where, $u_k \in \mathbb{R}$ is the control input for the $k^\text{th}$ agent, which acts in a direction lateral to the motion of the agent, thus controlling the curvature of the trajectory. A positive (resp., negative) value of $u_k$ corresponds to anticlockwise (resp., clockwise) rotation, while $u_k = 0$ corresponds to straight line motion in the initial velocity direction $\theta_k(0)$. Owing to the curvature dependence, our approach relies on designing $\zeta_k \in \mathbb{R}$ such that
\begin{equation}\label{eq_new_control}
u_k = \kappa(\phi)(1 + \zeta_k),
\end{equation} 
where, $\kappa(\phi)$ is the curvature of the desired (simple closed) polar curve parameterized by $\phi \in [0, 2\pi)$, (a detailed discussion about the curve's parameterization is given in Section~\ref{section3}). Here, $\zeta_k$ is derived from the formation control objectives of the group and approaches zero in the steady-state. Since the lateral force applied by an autonomous vehicle is often restricted due to its physical constraints, we further consider that the control $u_k$ in \eqref{eq_new_control} is given by the following saturation function
\begin{subequations}\label{saturated_control}
	\begin{numcases}
	{u_k \coloneqq}
	\label{saturation1} \kappa(\phi)(1 + \zeta_k), & if~$\left|\kappa(\phi)(1 + \zeta_k)\right| \leq u_\text{max}$\\
	\label{saturation2} u_\text{max}~\text{sgn}(\kappa(\phi)(1 + \zeta_k)),  & if~$\left|\kappa(\phi)(1 + \zeta_k)\right| > u_\text{max}$,
	\end{numcases}
\end{subequations}
where, $\text{sgn}(x)$ is the signum function of $x \in \mathbb{R}$, and $u_\text{max} > 0$ is the \emph{pre-specified} maximum allowable control force. Note that the saturation \eqref{saturation2} is applied only if $\kappa(\phi) \neq 0$; if $\kappa(\phi) = 0$ for some $\phi$ (i.e., the agent moves along a straight line), $u_k = 0$, which always satisfies the condition \eqref{saturation1}. Unless otherwise stated, we assume that $u_\text{max}$ is at least equal to the input demanded by the desired curve, that is, $u_{\max} \geq \max_{\phi} |\kappa(\phi)|$ for all $k$. In \eqref{modelNew}, $\dot{\theta}_k = u_k$ is usually referred to as the phase control model, which essentially controls the turn rates $\dot{\theta}_k$ of the agents, and hence, their phase angles $\theta_k$. In the next section, we propose a curve-phase model, a generalization of the phase-control model, to stabilize agents' motion around simple closed polar curves.  

\subsection{Barrier Lyapunov Function}
\begin{defn}[\hspace{-.1pt}Barrier Lyapunov Function \cite{tee2009barrier}]\label{BLF}
	A Barrier Lyapunov Function is a scalar function $V(\pmb{x})$ of state vector $\pmb{x} \in \mathcal{D}$ of the system $\dot{\pmb{x}} = f(\pmb{x})$ on an open region $\mathcal{D}$ containing the origin, that is continuous, positive definite, has continuous first-order partial derivatives at every point of $\mathcal{D}$, has the property $V(\pmb{x}) \rightarrow \infty$ as $\pmb{x}$ approaches the boundary of $\mathcal{D}$, and satisfies $V(\pmb{x}(t)) \leq \beta, \forall t\geq 0$, along the solution of $\dot{\pmb{x}} = f(\pmb{x})$ for $\pmb{x}(0) \in \mathcal{D}$ and some positive constant $\beta$.
\end{defn}

\begin{lem}[\hspace{-.1pt}\cite{tee2009barrier}]\label{lem1}
	For any positive constant $c$, let $\mathcal{Z} \coloneqq \{\xi \in \mathbb{R} \mid -c < \xi < c\} \subset \mathbb{R}$ and $\mathcal{N} \coloneqq \mathbb{R}^\ell  \times \mathcal{Z} \subset \mathbb{R}^{\ell+1}$ be open sets. Consider the system ${\dot{\pmb{\eta}} = \pmb{h}(t, \pmb{\eta})}$,  where, $\pmb{\eta} \coloneqq [\pmb{\tau}, \xi]^T \in \mathcal{N}$, and ${\pmb{h} : \mathbb{R}_{\geq 0} \times \mathcal{N}} \rightarrow \mathbb{R}^{\ell+1}$ is piecewise continuous in $t$ and locally Lipschitz in $\pmb{\eta}$, uniformly in $t$, on $\mathbb{R}_{\geq 0} \times \mathcal{N}$. Suppose that there exist functions ${U: \mathbb{R}^\ell \rightarrow \mathbb{R}_{\geq 0}}$ and $V_1: \mathcal{Z} \rightarrow \mathbb{R}_{\geq 0}$, continuously differentiable and positive definite in their respective domains, such that $ V_1(\xi) \rightarrow \infty~~~\text{as}~~~|\xi| \rightarrow c$ and $\gamma_1(\|\pmb{\tau}\|) \leq U(\pmb{\tau}) \leq \gamma_2(\|\pmb{\tau}\|)$, where, $\gamma_1$ and $\gamma_2$ are class $\mathcal{K}_\infty$ functions. Let $V(\pmb{\eta}) \coloneqq V_1(\xi) + U(\pmb{\tau})$, and $\xi(0) \in \mathcal{Z}$. If it holds that $\dot{V} = (\nabla V)^T{\pmb{h}} \leq 0$, in the set $\xi \in \mathcal{Z}$, then $\xi(t) \in \mathcal{Z},~\forall t \in [0, \infty)$.
\end{lem}

In the sequel, we use Lemma~\ref{lem1} to prove some theoretical results in this paper.

\begin{figure}
	\centering
	\includegraphics[width=0.5\columnwidth]{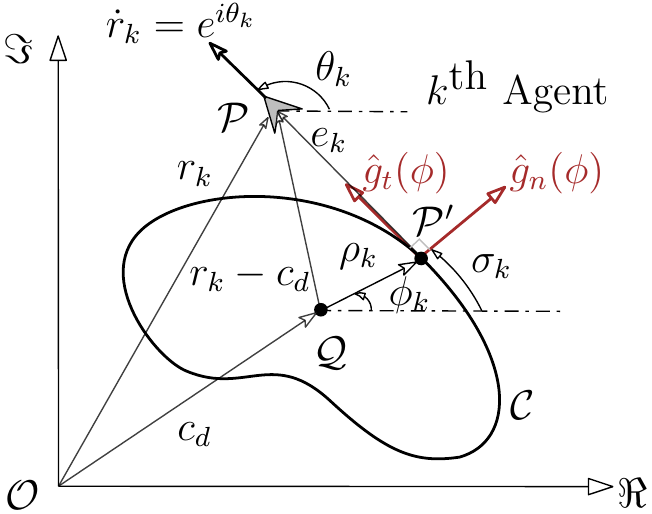}
	\caption{The $k^\text{th}$ agent tracking a simple closed curve $\mathcal{C}$ in the complex plane. Note that the position $\rho_k$ of the tracking point $\mathcal{P}'$ is a function of $\theta_k$.}
	\label{problem_fig}
\end{figure}

\section{Curvature Control and Curve-Phase Synchronization and Balancing}\label{section3}
This section develops a curve-phase control model for the agents' motion around smooth closed curves, and describes synchronized and balanced curve-phase patterns in their collective motion. 

\subsection{Curvature control and curve-phase model}\label{curve_def}
Our first goal is to allow the agents to move around the desired curve, characterized by a family of simple closed curves expressed in polar form. The problem is shown in Fig.~\ref{problem_fig}, where the $k^\text{th}$ agent is trying to move along the curve $\mathcal{C}$, centered at the desired location $c_d$. Consider that $\mathcal{C}$ is parameterized by $\phi$ with respect to its center $c_d$, and is represented by the map $\rho : [0,2\pi) \to \mathbb{C}, \phi \mapsto \rho(\phi)$. The unit tangent to $\mathcal{C}$ at the point $\rho(\phi)$ is $\hat{g}_{t}(\phi) = ({1}/{|{d\rho}/{d\phi}|})({d\rho}/{d\phi}) = {\rm e}^{i\mu} \in \mathbb{C}$, where  $\mu = \arg(\hat{g}_t)$. By rotating ${\rm e}^{i\mu}$ by an angle $\pi/2$ radians in the clockwise direction, we get the exterior unit normal $\hat{g}_{n}(\phi) = -i{\rm e}^{i\mu} \in \mathbb{C}$ (see Fig.~\ref{problem_fig}). The arc length along the curve at the point $\rho(\phi)$ is $\sigma:[0,2\pi) \to \mathbb{R}_{\geq0}, \phi \mapsto \sigma(\phi)$, and is given by
\begin{equation}\label{arc_length}
\sigma(\phi) = \int_{0}^{\phi} \left|\frac{d\rho}{d\bar{\phi}}\right| d\bar{\phi}.
\end{equation}
The curvature $\kappa:[0,2\pi) \to\mathbb{R}, \phi \mapsto \kappa(\phi)$ at the point $\rho(\phi)$ on the curve is $\kappa(\phi) = {d\mu}/{d\sigma} = ({d\mu}/{d\phi})({d\phi}/{d\sigma})$, which using \eqref{arc_length}, gives $\kappa(\phi) = ({1}/{|{d\rho}/{d\phi}|})({d\mu}/{d\phi})$. The sign of $\kappa(\phi)$ is determined by the sense of rotation around $\mathcal{C}$; if the curve is turning anticlockwise (resp., clockwise) at $\phi$, $\kappa(\phi) > 0~(\text{resp.}, \kappa(\phi) < 0)$. Note that this convention on $\kappa$ is given with reference to Fig.~\ref{problem_fig}, where the agent is moving in the anticlockwise direction. However, if the agent moves in the clockwise direction, an opposite convention holds as $\hat{g}_{n}(\phi)$ reverses its direction. For simple closed curves, the curvature is finite and bounded, that is, $0 \leq |\kappa(\phi)| < \infty$ for all $\phi$. 

\begin{lem}[\hspace{-.1pt}\cite{tapp2016differential}, pg.~40]\label{lem_curvature_polar_curve}
	Consider a family of simple closed curves, parameterized by $\phi$, and expressed in polar representation as $\rho(\phi) = R(\phi){\rm e}^{i\phi} \in \mathbb{C}, R(\phi) \in \mathbb{R}_{>0}$. Let $f(\phi), \sigma(\phi)$ and $\kappa(\phi)$ be the slope of tangent, arc-length and curvature, respectively. Then, $f(\phi) = (R'(\phi)\sin\phi + R(\phi)\cos\phi)/(R'(\phi)\cos\phi - R(\phi)\sin\phi), \sigma(\phi) = \int_{0}^{\phi} \sqrt{(R'(\bar{\phi}))^2 + (R(\bar{\phi}))^2} d\bar{\phi}$, and $\kappa(\phi) = (2(R'(\phi))^2 -R(\phi)R''(\phi) + (R(\phi))^2)/((R'(\phi))^2 + (R(\phi))^2)^{\frac{3}{2}}$, where, $R'(\phi) = dR/d\phi$, and $R''(\phi) = d^2R/d\phi^2$.  	
\end{lem}

From Lemma~\ref{lem_curvature_polar_curve}, one can deduce that 
\begin{equation}\label{eq_ratio}
\frac{1+f^2(\phi)}{f'(\phi)} = \frac{1}{\kappa(\phi)\sqrt{(R'(\phi))^2 + (R(\phi))^2}},
\end{equation}
where $f'(\phi) = df(\phi)/d\phi =  (2(R'(\phi))^2 -R(\phi)R''(\phi) + (R(\phi))^2)/(R'(\phi)\cos\phi - R(\phi)\sin\phi)^2$.

\begin{figure}
	\centering
	\subfigure[Convex lima\c{c}on]{\includegraphics[width=4.0cm]{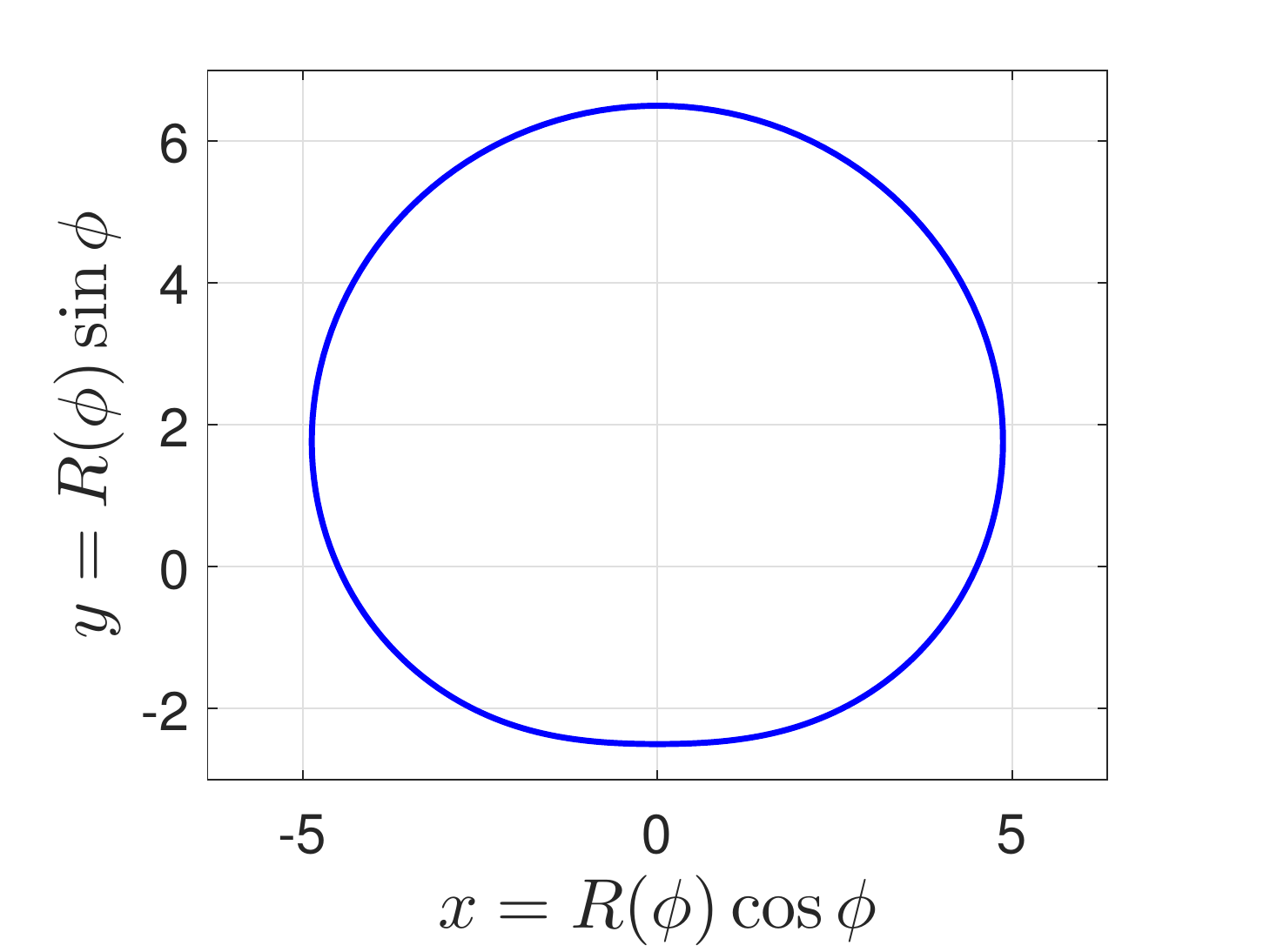}}
	\subfigure[Polar rose]{\includegraphics[width=4.0cm]{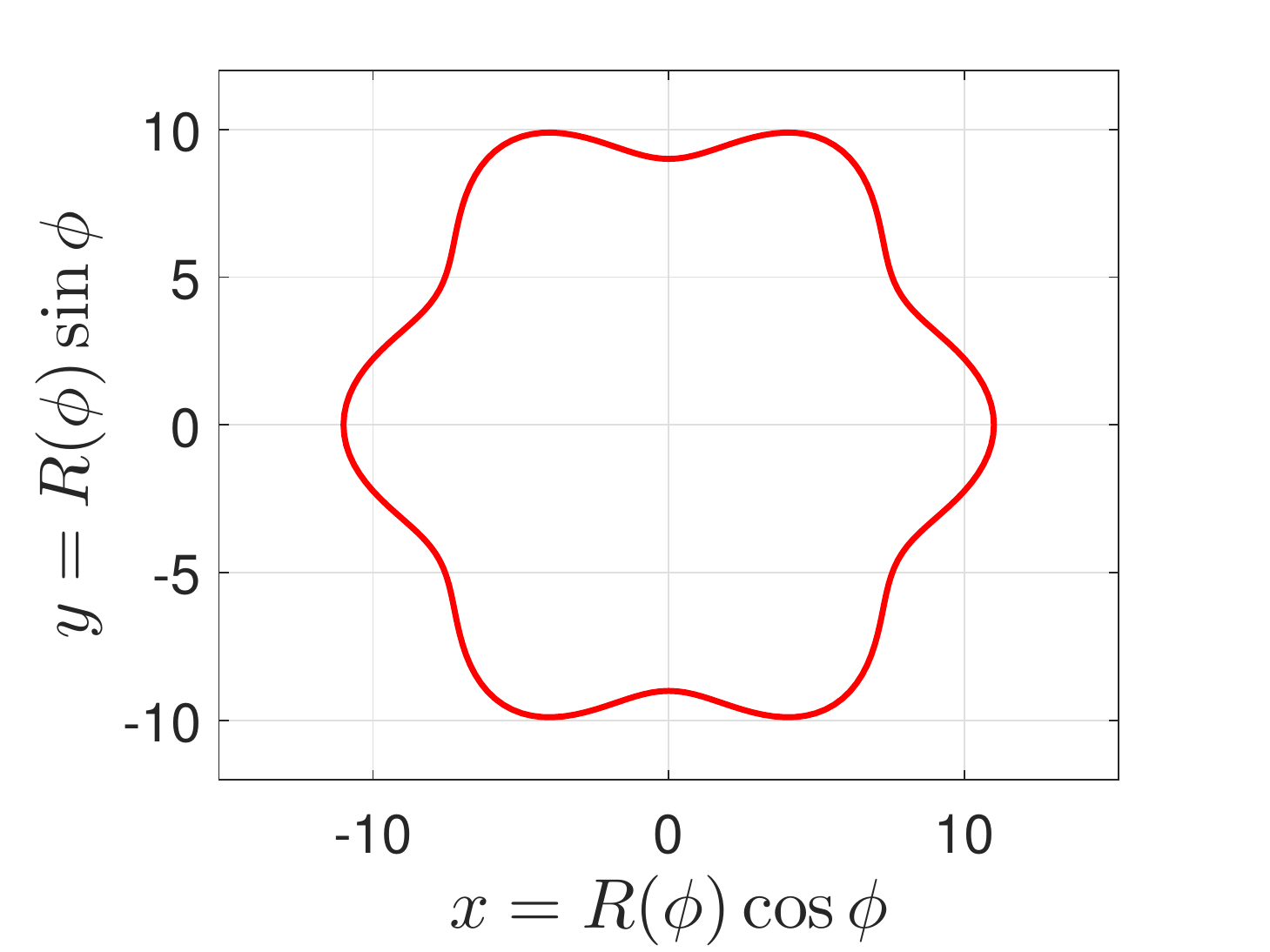}}
	\subfigure[$\kappa(\phi)-$Convex lima\c{c}on]{\includegraphics[width=4.0cm]{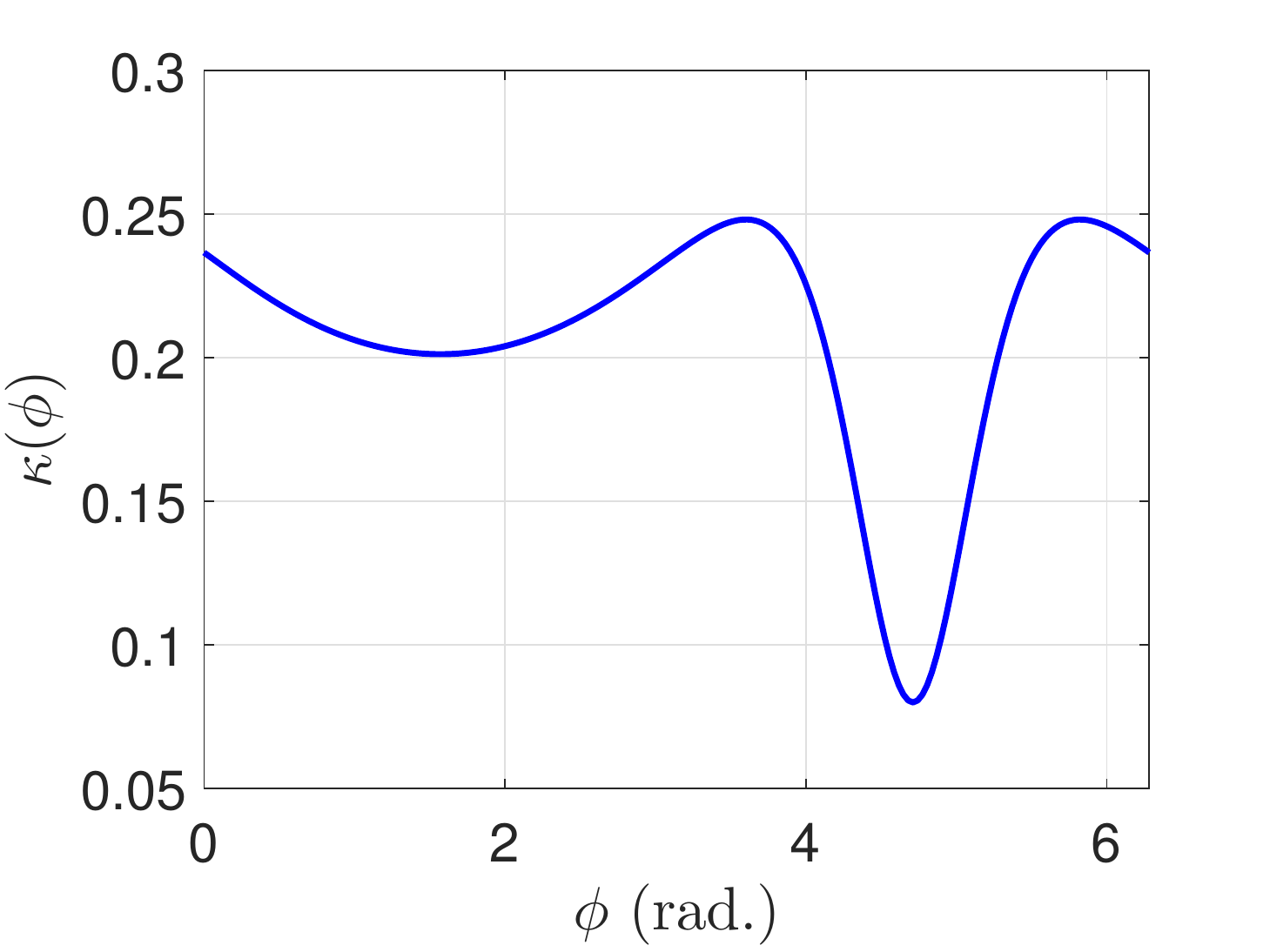}}
	\subfigure[$\kappa(\phi)-$Polar rose]{\includegraphics[width=4.0cm]{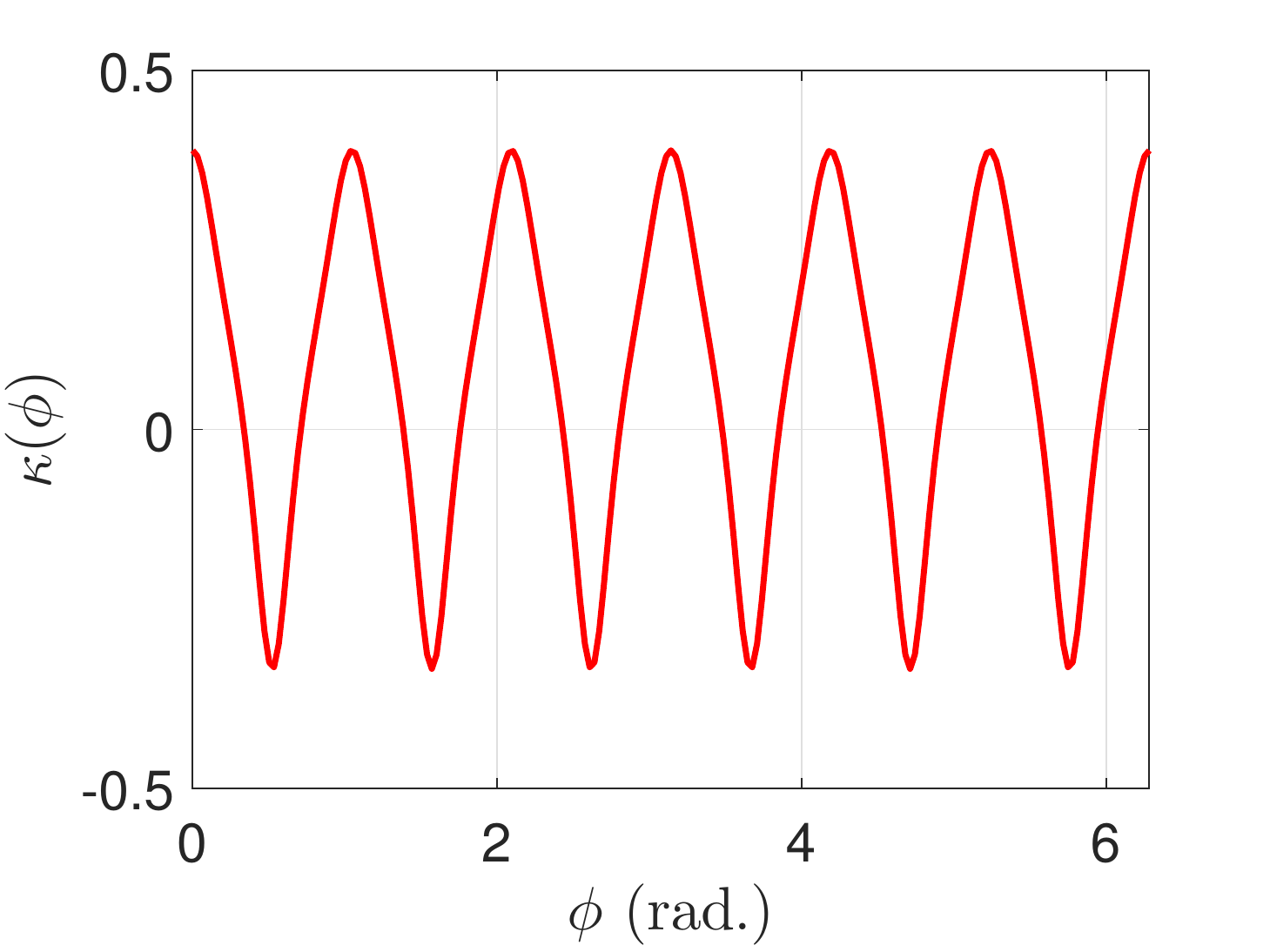}}
	\caption{Examples of simple closed curves and their curvature. Convex lima\c{c}on is plotted for $\hat{a} = 2, \hat{b} = 4.5$, and polar rose (non-convex) for $\tilde{a} = 10, \tilde{b} = 6, \tilde{s} = 1$.}
	\label{curves_fig}
\end{figure}

In the following example, we compare convex and non-convex polar curves and illustrate the challenges in designing the control laws for non-convex curves. 

\begin{example}\label{example1}
	Consider a family of simple closed curves, parameterized by $\phi$ with respect to the origin, and expressed in polar representation as $\rho(\phi) = R(\phi){\rm e}^{i\phi} = R(\phi)\cos\phi + iR(\phi)\sin\phi \in \mathbb{C}, R(\phi) \in \mathbb{R}_{>0}$. Depending upon $R(\phi)$, we illustrate the following two cases: \par 
	{Case~1:} A convex lima\c{c}on is an epitrochoid of the form $R(\phi) = \hat{b} + \hat{a}\sin\phi,~\hat{b} \geq 2\hat{a}$, where the condition $\hat{b} \geq 2\hat{a}$ ensures that the curve is simple and convex. One can obtain ${d\rho}/{d\phi} = \left(\hat{a}\cos 2\phi - \hat{b}\sin\phi\right) + i\left(\hat{a}\sin 2\phi + \hat{b}\cos\phi\right)$, which is essentially the tangent vector with slope $\tan\mu = ({\hat{a}\sin 2\phi + \hat{b}\cos\phi})/({\hat{a}\cos 2\phi - \hat{b}\sin\phi})$. From this, it can be obtained that ${d\mu}/{d\phi} = {\hat{\kappa}_N}/{\hat{\kappa}_D}$, where
	\begin{align*}
	\hat{\kappa}_N &= \hat{b}^2 + 2\hat{a}^2 + 3\hat{a}\hat{b}\sin\phi,\\
	\hat{\kappa}_D &= (\hat{a}\cos2\phi - \hat{b}\sin\phi)^2 + (\hat{a}\sin2\phi + \hat{b}\cos\phi)^2,
	\end{align*}
	which on substitution yields the curvature 
	\begin{equation*}
	\kappa(\phi) = \frac{d\mu}{d\phi}\frac{d\phi}{d\sigma} = \frac{\hat{\kappa}_N}{\hat{\kappa}_D}\frac{d\phi}{d\sigma} = \frac{\hat{\kappa}_N}{(\hat{\kappa}_D)^{\frac{3}{2}}},
	\end{equation*}
	where we have used the relation ${d\phi}/{d\sigma} = {1}/{|{d\rho}/{d\phi}|} = {1}/{\sqrt{\hat{\kappa}_D}}$, in the spirit of \eqref{arc_length}. \par 
	\par 
	{Case~2:} A polar rose is of the form $R(\phi) = \tilde{s}(\tilde{a}+ \cos(\tilde{b}\phi)), ~\tilde{a},\tilde{b}, \tilde{s} > 0$, where $\tilde{s}$ is a scaling factor and the condition $\tilde{a},\tilde{b},\tilde{s} > 0$ ensures that the curve is simple and closed. Similar to the previous case, one can obtain the tangent vector as ${d\rho}/{d\phi} = -\tilde{s}(\tilde{b}\sin(\tilde{b}\phi)\cos\phi + \tilde{a}\sin\phi + \cos(\tilde{b}\phi)\sin\phi) + i\tilde{s}(\cos(\tilde{b}\phi)\cos\phi + \tilde{a}\cos\phi - \tilde{b}\sin(\tilde{b}\phi)\sin\phi)$, which has slope $\tan\mu = -(\cos(\tilde{b}\phi)\cos\phi + \tilde{a}\cos\phi - \tilde{b}\sin(\tilde{b}\phi)\sin\phi)/(\tilde{b}\sin(\tilde{b}\phi)\cos\phi + \tilde{a}\sin\phi + \cos(\tilde{b}\phi)\sin\phi)$. It can be shown that ${d\mu}/{d\phi} = {\tilde{\kappa}_N}/{\tilde{\kappa}_D}$, where
	\begin{align*}
	\tilde{\kappa}_N &= \tilde{a}^2 + \tilde{b}^2 + \tilde{a}(\tilde{b}^2 + 2)\cos(\tilde{b}\phi) + \tilde{b}^2\sin^2(\tilde{b}\phi) + \cos^2(\tilde{b}\phi)\\
	\tilde{\kappa}_D &= (\tilde{b}\sin(\tilde{b}\phi)\cos\phi + \tilde{a}\sin\phi + \cos(\tilde{b}\phi)\sin\phi)^2\\
	& + (\cos(\tilde{b}\phi)\cos\phi + \tilde{a}\cos\phi - \tilde{b}\sin(\tilde{b}\phi)\sin\phi)^2, 
	\end{align*}
	and hence, the curvature is $\kappa(\phi) = (1/\tilde{s}){\tilde{\kappa}_N}/{(\tilde{\kappa}_D)^{\frac{3}{2}}}$, using similar steps as above. \par
	These curves, along with their curvature, are plotted in Fig.~\ref{curves_fig}. It is clear that the first curve (Fig.~\ref{curves_fig}(a)) is convex and the second (Fig.~\ref{curves_fig}(b)) is non-convex, and their curvatures are smooth and bounded. As the convexity of the curves change, the curvature changes its sign according to the convention mentioned above. Another important plot is shown in Fig.~\ref{many_to_one}, where $\phi$ is plotted against the angle $\mu$ of the tangent line to the curve. It is clear that $\phi(\mu)$ is a smooth bijective map for the convex lima\c{c}on, while this is not true for the polar rose, which is non-convex. Since the motion around a curve requires an agent to have the tangent velocity vector, non-uniqueness of $\phi$ to a tangent poses a challenge in stabilizing the motion of agents about a general class of simple closed curves including both convex and non-convex curves. A remedy for this is presented in Remark~\ref{remark_parameterization}.
\end{example}

\begin{figure}
	\centering
	\subfigure[$\phi(\mu)-$Convex lima\c{c}on]{\includegraphics[width=4.0cm]{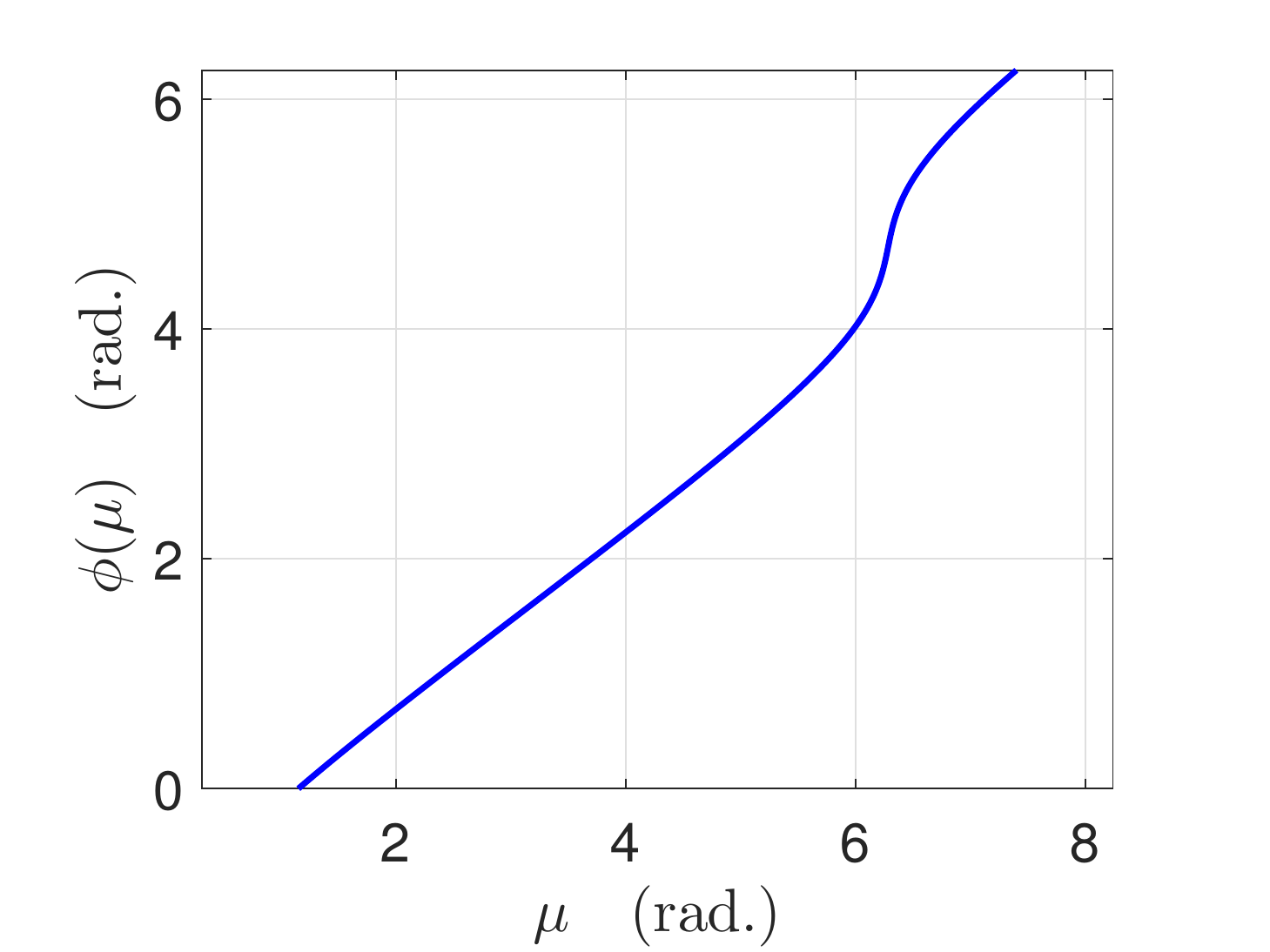}}
	\subfigure[$\phi(\mu)-$Polar rose]{\includegraphics[width=4.0cm]{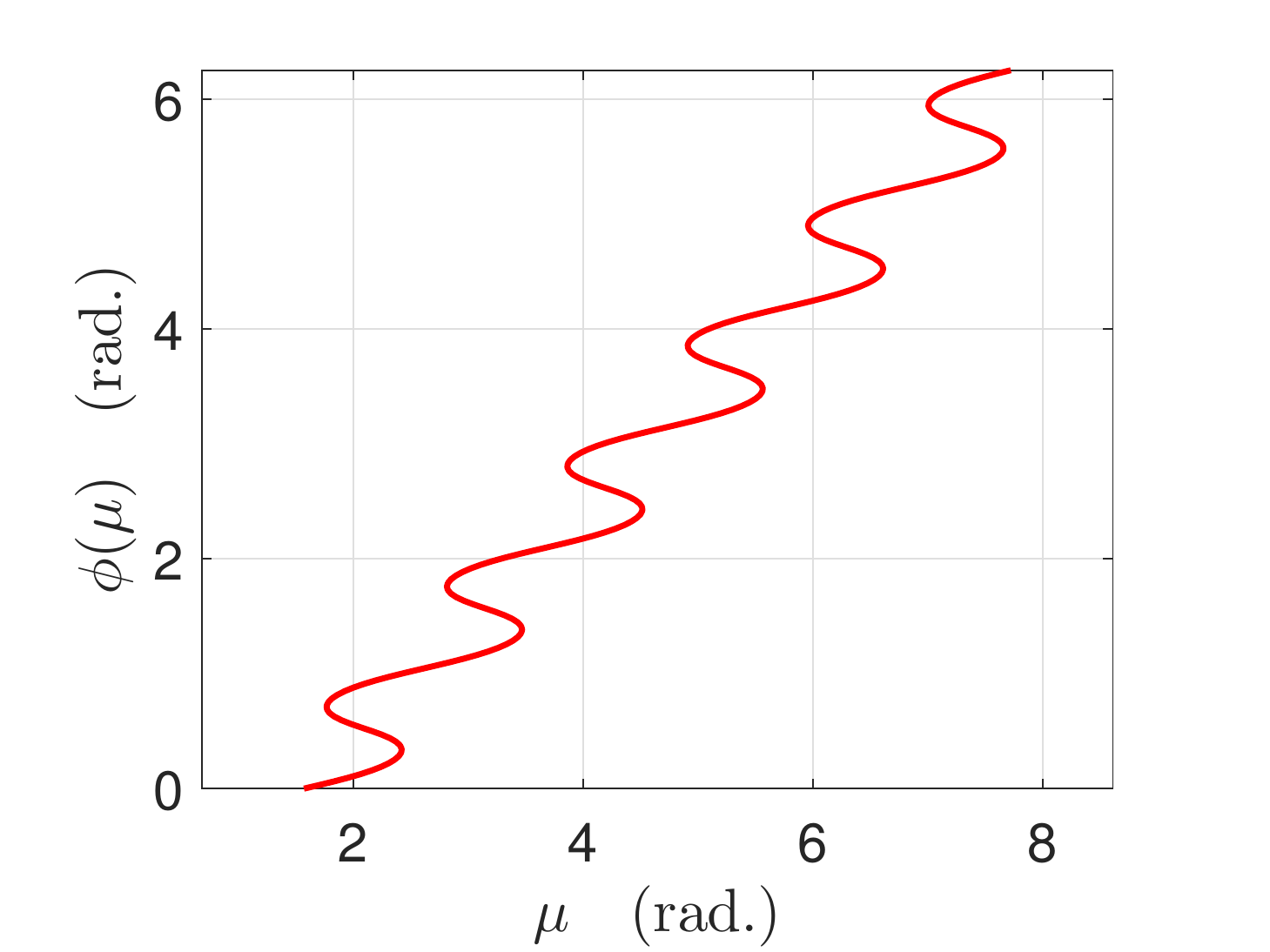}}
	\caption{The plots for $\phi(\mu)$ vs $\mu$ for the curves in Fig.~\ref{curves_fig}.}
	\label{many_to_one}
\end{figure}


Let us now turn our focus to the $k^\text{th}$ agent at point $\mathcal{P}$ in Fig.~\ref{problem_fig}, trying to move around the desired curve $\mathcal{C}$. Let $\rho_k \coloneqq \rho(\phi(\theta_k))$ be the required tracking point (the point $\mathcal{P}'$ in Fig.~\ref{problem_fig}) on $\mathcal{C}$ for the $k^\text{th}$ agent, associated to its heading $\theta_k$ by the smooth re-parametrization $\phi: \mathbb{S}^1 \to [0,2\pi), \theta_k \mapsto  \phi(\theta_k)$. In order to allow the $k^\text{th}$ agent to move around $\mathcal{C}$, the velocity constraint $\hat{g}_{t}(\phi(\theta_k)) = {\rm e}^{i\theta_k} \in \mathbb{C}$ must be satisfied. This is equivalent to the constraints $\tan\theta_k = f(\phi_k)$, where, $\phi_k \coloneqq \phi(\theta_k)$, and $f(\phi_k) = {\Im(\hat{g}_t(\phi_k))}/{\Re(\hat{g}_t(\phi_k))}$. The time derivative of $\tan\theta_k = f(\phi_k)$ leads to the following parametric-phase model
\begin{equation}\label{parametric-phase_model_intial}
\dot{\phi}_k = \frac{d\phi_k}{d\theta_k}\dot{\theta}_k = \frac{1+f^2(\phi_k)}{f'(\phi_k)}\dot{\theta}_k,~~\hat{g}_t(\phi_k(0)) = {\rm e}^{i\theta_k(0)},
\end{equation}  
which describes the evolution of point $\mathcal{P}'$ along the curve $\mathcal{C}$. Using \eqref{eq_new_control} and \eqref{eq_ratio}, a simplified parametric-phase model is obtained as
\begin{equation}\label{parameteric_phase_model}
\dot{\phi}_k = \frac{1 + \zeta_k}{\sqrt{(R'(\phi_k))^2 + (R(\phi_k))^2}} ,~~\hat{g}_t(\phi_k(0)) = {\rm e}^{i\theta_k(0)},
\end{equation}  
which is bounded for simple closed polar curves as $R(\phi_k) \in \mathbb{R}_{>0}$. 


\begin{remark}\label{remark_parameterization}
	Unlike convex curves, $\theta_k \mapsto \phi_k$ may not be a bijective map for non-convex curves, and hence, there may exist multiple values of $\phi_k$ for a $\theta_k$, as illustrated in Example~\ref{example1}. Among these, any value of $\phi_k$ may be chosen, provided $\hat{g}_t(\phi_k(0)) = {\rm e}^{i\theta_k(0)}$ and $|e_k(0)| < \delta$ for all $k$ (see below Theorem~\ref{theorem1}). This initialization, together with the parametric phase model \eqref{parameteric_phase_model}, gives a unique parametrization of $\phi_k$ with $\theta_k$.
\end{remark}


Along with stabilizing the agents' motion around the desired curve $\mathcal{C}$, we also achieve synchronized and balanced curve-phase patterns in their collective motion. In this direction, we define the curve-phase $\psi_k \coloneqq \psi(\phi_k)$ at a point $\rho_k$ on the curve $\mathcal{C}$ as follows \cite{paley2008stabilization}: 
\begin{equation}\label{curve_phase}
\psi_k = \frac{2\pi}{\Gamma_{\mathcal{C}}}\sigma_k,~k= 1, \ldots, N,
\end{equation}
where, $\sigma_k \coloneqq \sigma(\phi_k)$, as defined in \eqref{arc_length}, and $\Gamma_{\mathcal{C}} = \sigma(2\pi)$ is the perimeter of $\mathcal{C}$. The time derivative of \eqref{curve_phase}, along the dynamics \eqref{parameteric_phase_model}, yields the following curve-phase model
\begin{equation}\label{curve_phase_model}
\dot{\psi}_k  = \frac{2\pi}{\Gamma_{\mathcal{C}}}\frac{d\sigma_k}{dt} = \frac{2\pi}{\Gamma_{\mathcal{C}}}\frac{d\sigma_k}{d\phi_k}\dot{\phi}_k = \frac{2\pi}{\Gamma_{\mathcal{C}}}(1 + \zeta_k),
\end{equation}
where we used ${d\sigma_k}/{d\phi_k}= \sqrt{(R'(\phi_k))^2 + (R(\phi_k))^2}$ from Lemma~\ref{lem_curvature_polar_curve}. The next subsection describes curve-phase synchronization and balancing characterized by the curve-phases $\psi_k$. 


\subsection{Curve-Phase Synchronization and Balancing}
The curve-phase synchronization and balancing, around the desired curve $\mathcal{C}$, are characterized by the quantity $p_{\psi} \coloneqq ({1}/{N}) \sum_{k=1}^{N} {\rm e}^{i\psi_k} = |p_{\psi}|{\rm e}^{i\Psi}$, where, $|p_{\psi}|$ is its magnitude, and $\Psi$ is the resultant phase of the phasors ${\rm e}^{i\psi_k}$. The magnitude $|p_{\psi}|$ satisfies $0 \leq |p_{\psi}| \leq 1$, and is a measure of synchrony of $\pmb{\psi} = [\psi_1,\ldots,\psi_N]^T$. If $\psi_1 = \cdots = \psi_N$, then $\pmb{\psi}$ is synchronized and hence $|p_\psi|$ achieves its maximum value, that is, $|p_\psi| = 1$. On the other hand, $\pmb{\psi}$ is said to be balanced if $p_\psi = 0$, that is, the phasors ${\rm e}^{i\psi_k}$ add up to zero. Unlike \cite{jain2019trajectory}, the notions of synchronization and balancing are generalized here and are defined in terms of curve-phases $\psi_k$ instead of heading angles $\theta_k$. For the special case of circular motion, it is evident that synchronization and balancing of $\psi_k$ is equivalent to that of $\theta_k$. 

We consider the following Laplacian-based curve-phase potential function of phasors ${\rm e}^{i\psi_k}$ to stabilize synchronization and balancing around $\mathcal{C}$: 
\begin{equation}
\label{W}\mathcal{W}(\pmb{\psi}) = \frac{1}{2}\langle {\rm e}^{i\pmb{\psi}}, \mathcal{L}{\rm e}^{i\pmb{\psi}}\rangle,
\end{equation}
where, $ {\rm e}^{i\pmb{\psi}} = [{\rm e}^{i\psi_1}, \ldots, {\rm e}^{i\psi_N}]^T$, and $\mathcal{L}$ is the Laplacian of the underlying interaction topology. Since $\mathcal{L} = \mathcal{M}\mathcal{M}^T$ ($\mathcal{M}$ being the incidence matrix) for an undirected and connected graph, we have
\begin{align}\label{W_max}
\nonumber &\hspace*{-0.3cm}\frac{1}{2}\langle {\rm e}^{i\pmb{\psi}}, \mathcal{L}{\rm e}^{i\pmb{\psi}}\rangle = \frac{1}{2}\langle {\rm e}^{i\pmb{\psi}}, \mathcal{M}\mathcal{M}^T{\rm e}^{i\pmb{\psi}}\rangle = \frac{1}{2}\langle \mathcal{M}^T{\rm e}^{i\pmb{\psi}}, \mathcal{M}^T{\rm e}^{i\pmb{\psi}}\rangle = \\
& \hspace*{-0.3cm}  \frac{1}{2}\sum_{\{j,k\} \in \mathcal{E}}|{\rm e}^{i\psi_j} - {\rm e}^{i\psi_k}|^2 \leq \frac{1}{2}\sum_{\{j,k\} \in \mathcal{E}}(|{\rm e}^{i\psi_j}| + |{\rm e}^{i\psi_k}|)^2 = 2|\mathcal{E}|,
\end{align}
where, $|\mathcal{E}|$ is the cardinality of the edge set $\mathcal{E}$. 
\begin{lem}[\hspace{-.1pt}\cite{paley2008stabilization,jain2018collective}]\label{lem_critical points of W} 
	Let $\mathcal{L}$ be the Laplacian of an undirected and connected graph $\mathcal{G}$ with $N$ vertices. Consider the Laplacian curve-phase potential $\mathcal{W}(\pmb{\psi})$ defined in \eqref{W}. If ${\rm e}^{i\pmb{\psi}}$ is an eigenvector of $\mathcal{L}$, then $\pmb{\psi}$ is a critical point of $\mathcal{W}(\pmb{\psi})$, and $\pmb{\psi}$ is either synchronized or balanced. The potential $\mathcal{W}(\pmb{\psi})$ reaches its global minimum if and only if $\pmb{\psi}$ is synchronized. If $\mathcal{G}$ is circulant, then $\mathcal{W}(\pmb{\psi})$ reaches its global maximum in a balanced curve-phase arrangement.
\end{lem}

The proof of Lemma~\ref{lem_critical points of W} directly follows from Lemma~\ref{lem_circulant_matrices} \cite{jain2018collective}. From \eqref{W}, it is clear that, for an undirected and connected graph $\mathcal{G}$,  $\mathcal{W}(\pmb{\psi})$ achieves its minimum value zero when ${\rm e}^{i\pmb{\psi}} = {\rm e}^{i\psi_0}\pmb{1}_N$ for any $\psi_0 \in \mathbb{S}^1$, implying that the curve-phases $\pmb{\psi}$ are in synchronization. If the graph $\mathcal{G}$ is circulant, then Lemma~\ref{lem_circulant_matrices} allows writing $\mathcal{W}(\pmb{\psi}) = \frac{1}{2}\langle \mathcal{F}^\ast {\rm e}^{i\pmb{\psi}}, \Lambda \mathcal{F}^\ast {\rm e}^{i\pmb{\psi}}\rangle$, which on substituting $\pmb{w} = \mathcal{F}^\ast {\rm e}^{i\pmb{\psi}}$, yields $\mathcal{W}(\pmb{\psi}) = \frac{1}{2}\langle \pmb{w}, \Lambda\pmb{w}\rangle = \frac{1}{2}\sum_{k=2}^{N}|w_k|^2\lambda_k$. Since $\mathcal{F}$ is unitary, $\|\pmb{w}\| = \|{\rm e}^{i\pmb{\psi}}\| = \sqrt{N}$. Thus, $ \mathcal{W}(\pmb{\psi}) = (1/2)\left<\pmb{w}, \Lambda\pmb{w}\right> \leq (N/2) \lambda_{\text{max}}$, where, $\lambda_{\text{max}}$ is the maximum eigenvalue of $\mathcal{L}$. In other words, $\mathcal{W}(\pmb{\psi})$ is bounded by $(N/2) \lambda_{\text{max}}$ for a circulant graph $\mathcal{G}$, and the maximum value is achieved by selecting ${\rm e}^{i\pmb{\psi}}$ as the eigenvector of $\mathcal{L}$, associated with $\lambda_{\text{max}}$. Since ${\rm e}^{i\pmb{\psi}}$ is orthogonal to $\pmb{1}_N$, i.e., $\pmb{1}_N^T{\rm e}^{i\pmb{\psi}} = 0$, it corresponds to the balancing of curve-phases $\pmb{\psi}$.

The time-derivative of $\mathcal{W}(\pmb{\psi})$, along the curve-phase dynamics \eqref{curve_phase_model}, is
\begin{equation*}
\dot{\mathcal{W}} = \sum_{k=1}^{N} \left(\frac{\partial \mathcal{W}}{\partial \psi_k}\right) \dot{\psi}_k = \frac{2\pi}{\Gamma_{\mathcal{C}}}\sum_{k=1}^{N} \left(\frac{\partial \mathcal{W}}{\partial \psi_k}\right)(1 + \zeta_k).
\end{equation*}
The gradient ${\partial \mathcal{W}}/{\partial \psi_k}$ can be calculated as ${\partial \mathcal{W}}/{\partial \psi_k} = \langle i{\rm e}^{i\psi_k}, \mathcal{L}_k {\rm e}^{i\pmb{\psi}}\rangle = -\sum_{j \in \mathcal{N}_k} \sin(\psi_j - \psi_k)$, where, $\mathcal{L}_k$ is the $k^\text{th}$ row of the Laplacian $\mathcal{L}$. As a result,
\begin{equation}
\label{W_dot_final}\dot{\mathcal{W}} =  \frac{2\pi}{\Gamma_{\mathcal{C}}}\sum_{k=1}^{N} \langle i{\rm e}^{i\psi_k}, \mathcal{L}_k{\rm e}^{i\pmb{\psi}}\rangle (1 + \zeta_k), 
\end{equation}
and $\sum_{k=1}^{N} {\partial \mathcal{W}}/{\partial \psi_k} = \langle i{\rm e}^{i\pmb{\psi}}, \mathcal{M}\mathcal{M}^T {\rm e}^{i\pmb{\psi}}\rangle = \langle i\mathcal{M}^T{\rm e}^{i\pmb{\psi}}, \mathcal{M}^T{\rm e}^{i\pmb{\psi}}\rangle = 0$, implying that $\langle \nabla_{\pmb{\psi}} \mathcal{W}, \pmb{1}_N \rangle = 0$, i.e., $\nabla_{\pmb{\psi}} \mathcal{W}$ and $\pmb{1}_N$ are orthogonal. 

\section{Control Design}\label{section4}
This section derives feedback control laws that enforce the collective motion of the agents around the desired simple closed polar curve in synchronized or balanced curve-phase patterns. The proposed controllers also assure that the agents' trajectories remain bounded during stabilization in either of the phase patterns, and their turn-rates adhere to the desired saturation limits. From Fig.~\ref{problem_fig}, the error $e_k$ is given by $\mathcal{P}'\mathcal{P} = \mathcal{Q}\mathcal{P} - \mathcal{Q}\mathcal{P}'$, leading to
\begin{equation}\label{error_var}
e_k = r_k - c_d - \rho_k.
\end{equation}
Note that \eqref{error_var} is valid even if the agents move in the clockwise direction as the unit vector $\hat{g}_n(\phi)$ reverses its direction and the sign of curvature $\kappa_k$ also changes. Thus, without loss of generality, further analysis is carried out with respect to Fig.~\ref{problem_fig}. 

The time derivative of $e_k$, along dynamics \eqref{modelNew} and \eqref{parameteric_phase_model}, is
\begin{equation}\label{error_der}
\dot{e}_k = \dot{r}_k - \dot{\rho}_k = -{\rm e}^{i\theta_k}\zeta_k,
\end{equation}
where, $\dot{\rho}_k = ({d\rho_k}/{d\phi_k})\dot{\phi}_k = {\rm e}^{i\theta_k}(1 + \zeta_k)$ is used to simplify the expression, in conjunction with the relations, ${d\rho_k}/{d\phi_k} = (|{d\rho_k}/{d\phi_k}|){\rm e}^{i\theta_k}$, and $|{d\rho_k}/{d\phi_k}| = \sqrt{(R'(\phi_k))^2 + (R(\phi_k))^2}$ from Lemma~\ref{lem_curvature_polar_curve}. To allow the agents to move around $\mathcal{C}$, the error $e_k, \forall k$, is minimized by using the following logarithmic BLF-based collective potential function
\begin{equation}\label{BLF_potential_fn}
\hspace*{-0.1cm}\mathcal{S}(\pmb{e}) = \mathcal{S}(\pmb{r}, \pmb{\theta}) \coloneqq \sum_{k=1}^{N} \mathcal{S}_{k}(r_k,\theta_k) =\frac{1}{2} \sum_{k=1}^{N} \ln \left(\frac{\delta^2}{\delta^2 - |e_k|^2}\right),
\end{equation}
where, `$\ln$' denotes natural logarithm, $\delta > 0$ is a constant, $\pmb{e} = [e_1, \ldots, e_N]^T$ is the error vector, and $\mathcal{S}_{k}(r_k,\theta_k) = ({1}/{2}) \ln ({\delta^2}/{(\delta^2 - |e_k|^2)})$ is the BLF for the $k^\text{th}$ agent. The potential $\mathcal{S}(\pmb{e})$ is positive definite and continuously differentiable in the region $|e_k(t)| < \delta, \forall k$ \cite{tee2009barrier}, and is zero when $\pmb{e} = \pmb{0}_N$. Thus, the minimization of $\mathcal{S}(\pmb{e})$ corresponds to the collective motion around the desired curve $\mathcal{C}$, that is, $e_k = 0$ for all $k$ in \eqref{error_var}, implying that
\begin{equation}\label{position_new}
r_k = c_d + \rho_k,~\forall k,
\end{equation}
which is the position of the $k^\text{th}$ agent on curve $\mathcal{C}$. The time derivative of $\mathcal{S}(\pmb{e})$, along the dynamics \eqref{modelNew} and \eqref{parameteric_phase_model}, is $\dot{\mathcal{S}}  = \frac{1}{2}\sum_{k=1}^{N} ({{\frac{d}{dt}|e_k|^2}})/({\delta^2 - |e_k|^2})$, where, $\frac{1}{2}\frac{d}{dt}|e_k|^2 = \langle e_k, \dot{e}_k \rangle $. Substituting for $e_k$ and $\dot{e}_k$ from \eqref{error_var} and \eqref{error_der}, and simplifying the inner product, yields
\begin{equation}
\label{S_dot}\dot{\mathcal{S}} = -\sum_{k=1}^{N} \frac{\langle r_k-c_d-\rho_k, {\rm e}^{i\theta_k} \rangle}{\delta^2 - |e_k|^2}\zeta_k. 
\end{equation}
We now propose in the following theorem a Lyapunov-based framework to achieve curve-phase synchronization and balancing around $\mathcal{C}$, along with bounded trajectories and control saturation.   

\begin{thm}\label{theorem1}
	Let $\mathcal{L}$ be the Laplacian of an undirected and connected graph $\mathcal{G}$ with $N$ vertices. Consider the agent, parametric-phase, and curve-phase models \eqref{modelNew}, \eqref{parameteric_phase_model} and \eqref{curve_phase_model}, respectively. Assume that the initial states of the agents belong to the set $\mathcal{Z}_\delta \coloneqq \{(\pmb{r}, \pmb{\theta}) \in \mathbb{C}^N\times\mathbb{T}^N \mid |e_k| < \delta, ~\forall k\}$, where $e_k$ is defined in \eqref{error_var}, and $\delta > 0$ is a positive constant. Let the agents be governed by the saturated control law \eqref{saturated_control}, where
	\begin{equation}\label{control1}
	\zeta_k = K_{\mathcal{C}}\frac{\langle r_k-c_d-\rho_k, {\rm e}^{i\theta_k} \rangle}{\delta^2 - |e_k|^2} +  K\langle i{\rm e}^{i\psi_k}, \mathcal{L}_k{\rm e}^{i\pmb{\psi}}\rangle,
	\end{equation}
	for all $k = 1, \ldots, N$. Then, the following properties hold:
	\begin{enumerate}
	\item[i)] If $K_{\mathcal{C}} > 0$, and $K < 0$, all the agents asymptotically converge to the desired curve $\mathcal{C}$, centered at $c_d$, in a synchronized curve-phase arrangement in the set $\mathcal{Z}_\delta$.\par  
	\item[ii)] Additionally, if $\mathcal{G}$ is circulant with $K_{\mathcal{C}} > 0$, and $K > 0$, all the agents asymptotically converge to the desired curve $\mathcal{C}$, centered at $c_d$, in a balanced curve-phase arrangement in the set $\mathcal{Z}_\delta$. 
	\item[iii)] The trajectories of the agents in both the above cases stay within the set $\mathcal{B}_\delta = \bigcup_{\phi \in [0, 2\pi)}\mathcal{B}(\mathcal{C}(\phi), \delta)$ for all $k$ and $t \geq 0$, where $\mathcal{B}(\mathcal{C}(\phi), \delta) \coloneqq \{z \in \mathbb{C}~|~|z - c_d - \rho(\phi)| < \delta\}$ is the open disc of radius $\delta$ centered at $\mathcal{C}(\phi) = c_d + \rho(\phi)$ at $\phi$, and $\rho(\phi)$ is a parametrization of $\mathcal{C}$.	
	\end{enumerate}
\end{thm}

\begin{proof}
	\begin{enumerate}
	\item[i)] Consider the composite potential function
	\begin{equation}\label{V_1}
	\mathcal{V}_1(\pmb{r}, \pmb{\theta}) = K_{\mathcal{C}} \mathcal{S}(\pmb{r}, \pmb{\theta}) - K\frac{\Gamma_{\mathcal{C}}}{2\pi} \mathcal{W}(\pmb{\psi});~~~K_{\mathcal{C}}>0,~K<0,
	\end{equation}
	which is positive definite and bounded from below by zero. The time derivative of $\mathcal{V}_1(\pmb{r}, \pmb{\theta})$, along the dynamics \eqref{modelNew}, \eqref{parameteric_phase_model} and \eqref{curve_phase_model}, is $	\dot{\mathcal{V}}_1 = K_{\mathcal{C}}\dot{\mathcal{S}} -K ({\Gamma_{\mathcal{C}}}/{2\pi})\dot{\mathcal{W}}$. Substituting $\dot{\mathcal{W}}$ and $\dot{\mathcal{S}}$ from \eqref{W_dot_final} and \eqref{S_dot}, respectively, yields $\dot{\mathcal{V}}_1 = -\sum_{k=1}^{N} K_{\mathcal{C}}\frac{\langle r_k-c_d-\rho_k, {\rm e}^{i\theta_k} \rangle}{\delta^2 - |e_k|^2}\zeta_k - \sum_{k=1}^{N} K \langle i{\rm e}^{i\psi_k}, \mathcal{L}_k{\rm e}^{i\pmb{\psi}}\rangle (1+\zeta_k)$. Using the orthogonal property $\langle \nabla_{\pmb{\psi}} \mathcal{W}, \pmb{1}_N \rangle = 0$, we have that $\dot{\mathcal{V}}_1 = -\sum_{k=1}^{N} \left[K_{\mathcal{C}}\frac{\langle r_k-c_d-\rho_k, {\rm e}^{i\theta_k} \rangle}{\delta^2 - |e_k|^2} + K \langle i{\rm e}^{i\psi_k}, \mathcal{L}_k{\rm e}^{i\pmb{\psi}}\rangle\right]\zeta_k$. Under the control \eqref{control1}, this leads to $\dot{\mathcal{V}}_1 = -\sum_{k=1}^{N} \zeta_k^2 \leq 0$, along the closed loop solutions of \eqref{modelNew}. On the other hand, for the given saturation limit $u_{\max}$, it follows from \eqref{saturation2} that $|1+\zeta_k| > u_{\max}/|\kappa(\phi)|$, where $\kappa(\phi) \neq 0$, as discussed below Eq.~\eqref{saturated_control}. This implies that $\zeta_k < -(1 + u_{\max}/|\kappa(\phi)|) < -2$ or $\zeta_k > (u_{\max}/|\kappa(\phi)| - 1) > 0$, since $u_{\max} \geq \max_{\phi}|\kappa(\phi)|$. Therefore, $\dot{\mathcal{V}}_1$ is strictly less than zero in $\mathcal{Z}_{\delta}$, in case of saturation. 
	
	To account for both the scenarios collectively, we consider the general case when $\dot{\mathcal{V}}_1$ is negative semi-definite, i.e., $\dot{\mathcal{V}}_1 \leq 0$ for all $t\geq0$. This  implies that  $\mathcal{V}_1(\pmb{r}, \pmb{\theta})$ is non-increasing, that is, $\mathcal{V}_1(\pmb{r}, \pmb{\theta}) \leq \mathcal{V}_1(\pmb{r}(0), \pmb{\theta}(0)),~\forall t \geq 0$, along the solutions of system \eqref{modelNew} in $\mathcal{Z}_{\delta}$. Moreover, for every initial condition in $\mathcal{Z}_{\delta}$, it follows from \eqref{V_1} that $\mathcal{V}_1(\pmb{r}(0), \pmb{\theta}(0))\leq \beta = K_{\mathcal{C}}\sup_{(\pmb{r}(0), \pmb{\theta}(0)) \in \mathcal{Z}_\delta}\mathcal{S}(\pmb{r}(0), 
	\pmb{\theta}(0)) - K |\mathcal{E}|(\Gamma_{\mathcal{C}}/\pi)$, using \eqref{W_max}, where perimeter $\Gamma_{\mathcal{C}}$, control gains $K_{\mathcal{C}} > 0$, and $K < 0$ are finite. Thus, for the given $\delta$, ${\mathcal{V}}_1(\pmb{r}, \pmb{\theta})$ is bounded by a positive $\beta$ for all $t \geq 0$, along the solutions of \eqref{modelNew}, and has the property that ${\mathcal{V}}_1(\pmb{r}, \pmb{\theta}) \to \infty$, as its argument approaches the boundary $\partial\mathcal{Z}_\delta = \{(\pmb{r}, \pmb{\theta}) \in \mathbb{C}^N\times\mathbb{T}^N \mid |e_k| = \delta, ~\forall k\}$. Hence, ${\mathcal{V}}_1(\pmb{r}, \pmb{\theta})$ is a BLF for the set $\mathcal{Z}_\delta$, according to Definition~\ref{BLF}. 
	
	To prove convergence to the desired curve $\mathcal{C}$, note that the set $\Omega_{\beta} = 
	\{(\pmb{r}, \pmb{\theta}) \in \mathcal{Z}_\delta \mid {\mathcal{V}_1}(\pmb{r}, 
	\pmb{\theta}) \leq \beta\} \subset \mathcal{Z}_\delta$ is compact and 
	positively invariant, since ${\mathcal{V}_1}(\pmb{r}, \pmb{\theta})$ is 
	positive definite and continuously differentiable, and 
	$\dot{\mathcal{V}}_1 \leq 0$, along the solutions of \eqref{modelNew}, in $\mathcal{Z}_\delta$. 
	Therefore, it follows from LaSalle's invariance principle \cite{khalil2002nonlinear} that 
	all the solutions of system dynamics \eqref{modelNew}, under control 
	\eqref{control1}, converge to the largest invariant set $\Delta^s_{\mathcal{C}}$, contained 
	in the set $\Delta \subset \Omega_{\beta}$, where $\dot{\mathcal{V}}_1 = 0$. Thus,
	\begin{equation}
	\label{invariant_set}\Delta = \{(\pmb{r}, \pmb{\theta}) \in \mathcal{Z}_\delta 
	\mid \zeta_k = 0, \forall k\},
	\end{equation}
	which implies using \eqref{eq_new_control} that $u_k = \kappa(\phi), \forall k$ in $\Delta$. Further, it follows from \eqref{control1} that $K_{\mathcal{C}}\frac{\langle r_k-c_d-\rho_k, {\rm e}^{i\theta_k} \rangle}{\delta^2 - |e_k|^2} = -K \langle i{\rm e}^{i\psi_k}, \mathcal{L}_k{\rm e}^{i\pmb{\psi}}\rangle, \forall k$ in $\Delta$, which upon taking time-derivative on both the sides, yields $K_{\mathcal{C}} \frac{d}{dt} \left(\frac{\langle r_k-c_d - \rho_k, {\rm e}^{i\theta_k} \rangle }{\delta^2 - |e_k|^2}\right) = -K \frac{d}{dt} \left(\frac{\partial \mathcal{W}}{\partial \psi_k}\right)$, where $\frac{d}{dt}\left(\frac{\partial \mathcal{W}}{\partial \psi_k}\right) = \frac{d}{dt} \langle i{\rm e}^{i\psi_k}, \mathcal{L}_k{\rm e}^{i\pmb{\psi}}\rangle =  \langle \frac{d}{dt}(i{\rm e}^{i\psi_k}), \mathcal{L}_k{\rm e}^{i\pmb{\psi}}\rangle + \langle i{\rm e}^{i\psi_k}, \frac{d}{dt}\mathcal{L}_k{\rm e}^{i\pmb{\psi}}\rangle = \frac{2\pi}{\Gamma_{\mathcal{C}}}\left(\langle -{\rm e}^{i\psi_k}, \mathcal{L}_k{\rm e}^{i\pmb{\psi}}\rangle \right.$ $\left.+ \langle i{\rm e}^{i\psi_k}, i\mathcal{L}_k{\rm e}^{i\pmb{\psi}}\rangle\right) = 0$. Consequently, for all points in $\Delta^s_{\mathcal{C}} \subset \Delta$, we have $\frac{d}{dt} \frac{\langle r_k-c_d - \rho_k, {\rm e}^{i\theta_k} \rangle }{\delta^2 - |e_k|^2} = 0 \implies [(\delta^2 - |e_k|^2)\frac{d}{dt}{\langle r_k-c_d - \rho_k, {\rm e}^{i\theta_k}\rangle} - {\langle r_k-c_d - \rho_k, {\rm e}^{i\theta_k}\rangle}\frac{d}{dt}(\delta^2 - |e_k|^2)] = 0$.
	From \eqref{S_dot}, it is straightforward to see that $\frac{d}{dt}(\delta^2 - |e_k|^2) = 0$ in the set $\Delta^s_{\mathcal{C}}$. Thus, the previous expression reduces to $\frac{d}{dt}{\langle r_k-c_d - \rho_k, {\rm e}^{i\theta_k}\rangle} = 0 \implies \langle \dot{r}_k - \dot{\rho}_k, {\rm e}^{i\theta_k} \rangle + \kappa_k\langle r_k-c_d - \rho_k, i{\rm e}^{i\theta_k} \rangle = 0 \implies \langle r_k-c_d - \rho_k, i{\rm e}^{i\theta_k} \rangle = 0$, as $\dot{r}_k - \dot{\rho}_k = 0$ in $\Delta^s_{\mathcal{C}}$ (see \eqref{error_der}). For $\langle r_k-c_d - \rho_k, i{\rm e}^{i\theta_k} \rangle = 0$ to hold, it is necessary that $r_k = c_d + \rho_k$ (or $e_k \perp i{\rm e}^{i\theta_k}$ in $\Delta^s_{\mathcal{C}}$), which is the position of the $k^\text{th}$ agent moving around the curve $\mathcal{C}$ centered at $c_d$ (see \eqref{position_new}). Thus, every trajectory of \eqref{modelNew}, under control \eqref{control1}, approaches $\Delta^s_{\mathcal{C}}$ as $t \rightarrow \infty$, i.e., all the agents asymptotically converge to the desired curve $\mathcal{C}$ with center $c_d$ in $\mathcal{Z}_{\delta}$. Alternatively, $|e_k| \rightarrow 0$ as $t \rightarrow \infty$ and hence, $\mathcal{S}(\pmb{e})$ achieves its minimum if $|e_k(0)| < \delta, \forall k$, which follows from Lemma~\ref{lem1}, as each term $\mathcal{S}_k(r_k, \theta_k)$ of $\mathcal{S}(\pmb{r}, \pmb{\theta}) = \sum_{k=1}^{N}\mathcal{S}_k(r_k, \theta_k)$ is bounded and approaches zero. 
	
	Since $\mathcal{V}_1(\pmb{r}, \pmb{\theta}) = \mathcal{W}(\pmb{\psi})$ in $\Delta^s_{\mathcal{C}}$, we conclude that the agents reach an equilibrium with the asymptotic curve-phase arrangements in the critical set of $\mathcal{W}(\pmb{\psi})$. Since $\dot{\mathcal{V}}_1 \leq 0$, $\mathcal{W}(\pmb{\psi})$ also approaches zero, and hence, it follows from Lemma~\ref{lem_critical points of W} that the agents are in curve-phase synchronization around the desired curve $\mathcal{C}$. 
	\item[ii)] To prove this statement, let us consider the potential function
	\begin{equation}\label{V2}
	\mathcal{V}_2(\pmb{r}, \pmb{\theta}) = \mathcal{K}_{\mathcal{C}}\mathcal{S}(\pmb{r}, \pmb{\theta}) + K\frac{\Gamma_{\mathcal{C}}}{2\pi}\left(\frac{N}{2}\lambda_{\text{max}}-\mathcal{W}(\pmb{\psi})\right),
	\end{equation}
	with $\mathcal{K}_{\mathcal{C}}>0, K > 0$, which is a valid candidate as $0 \leq \mathcal{W}(\pmb{\psi}) \leq (N/2)\lambda_{\text{max}}$ for an undirected and connected circulant graph. The time derivative of $\mathcal{V}_2(\pmb{r}, \pmb{\theta})$ along dynamics \eqref{modelNew}, \eqref{parameteric_phase_model} and \eqref{curve_phase_model}, under control \eqref{control1}, results in $\dot{\mathcal{V}}_2 = -\sum_{k=1}^{N}\zeta_k^2  = \dot{\mathcal{V}}_1 \leq 0$. Thus, the rest of the proof follows the same steps as given for $K < 0$. However, in this case, let $\Delta^b_{\mathcal{C}}$ be the largest invariant set in $\Delta$ (defined in \eqref{invariant_set}), which every trajectory of \eqref{modelNew} approaches as $t \to \infty$. Following the above analysis, it can be concluded that all the agents asymptotically converge to the desired curve $\mathcal{C}$ with center $c_d$ in $\mathcal{Z}_\delta$. Moreover, $({N}/{2})\lambda_{\text{max}}-\mathcal{W}(\pmb{\psi})$ also approaches zero in $\mathcal{Z}_\delta$, as $\dot{\mathcal{V}}_2 \leq 0$. Since the graph $\mathcal{G}$ is circulant, it follows from Lemma~\ref{lem_critical points of W} that the agents are in curve-phase balancing around the desired curve $\mathcal{C}$ in $\mathcal{Z}_{\delta}$. \par
	
	
	\item[iii)] Since $\dot{\mathcal{V}}_1 = \dot{\mathcal{V}}_2 \leq 0$ in $\mathcal{Z}_\delta$, $\mathcal{S}(\pmb{e}) \left(= \sum_{k=1}^{N}\mathcal{S}_k\right)$ remains bounded. Consequently, for all $k = 1, \ldots, N$, $|e_k(\phi(t))| < \delta, \forall t \geq 0$, according to Lemma~\ref{lem1}. Substituting for $e_k$ from \eqref{error_var}, $|e_k(\phi(t))| < \delta \implies |r_k - c_d - \rho(\phi)| < \delta, \forall k$, and $\forall \phi \in [0, 2\pi)$. This, in turn, implies that the trajectories of the agents stay within the set $\mathcal{B}_\delta = \bigcup_{\phi \in [0, 2\pi)}\mathcal{B}(\mathcal{C}(\phi), \delta)$ for all $k$ and $t \geq 0$ in both synchronized and balanced curve-phase arrangements, where $\mathcal{B}(\mathcal{C}(\phi), \delta) \coloneqq \{z \in \mathbb{C}~|~|z - c_d - \rho(\phi)| < \delta\}$ is the open disc of radius $\delta$ and center $\mathcal{C}(\phi) = c_d + \rho(\phi)$ at $\phi$. Since the result follows for all $k$, the subscript $k$ is excluded from $\phi$ associated to the $k^\text{th}$ agent, and is directly related in terms of its parametrization $\rho(\phi)$.
	\end{enumerate}
\end{proof}

\begin{remark}
	In \eqref{control1}, one may infer that $\zeta_k$ becomes unbounded whenever $|e_k| = \delta$, due to the presence of the term $(\delta^2 - |e_k|^2)$ in the denominator. However, it has been established in Theorem~\ref{theorem1} that, in the closed loop, the error signal $|e_k| < \delta, \forall t \geq 0$, thereby $\zeta_k$, and hence, $u_k$ always remains finite for any solution trajectory. Further, note that $\zeta_k$ may assume any (large/small) value, depending upon the variables in \eqref{control1}. However, the actual applied control $u_k$ in \eqref{modelNew} is always bounded as per Eq.~\eqref{saturated_control}.  
\end{remark}

From the preceding discussion, one can observe that, by limiting the magnitude of the error variables $|e_k|$, the agents' trajectories $r_k$ remain bounded within the set $\mathcal{B}_{\delta}$, while there is no restriction on the heading angles $\theta_k$ of the agents and these act as the free variables, in accordance with Lemma~\ref{lem1}.

In general, it is hard to characterize the explicit nature of the boundary $\partial\mathcal{B}_{\delta}$ of $\mathcal{B}_{\delta}$ in Theorem~\ref{theorem1}. However, the following assumption on $\delta$ allows us to do so, as discussed in Corollary~\ref{corollary1} below. 
\begin{assumption}\label{Assumption1}
	There exists a constant $\delta > 0$ such that, for every $\epsilon \in (0, \delta]$, $c_d + \rho(\phi) \pm\epsilon\hat{g}_{n}(\phi) \not\in \mathcal{C}, \phi \in [0, 2\pi)$. Moreover, for $\phi_1, \phi_2 \in [0, 2\pi), \phi_1 \neq \phi_2$, it holds that $c_d + \rho(\phi_1) + \epsilon \hat{g}_{n}(\phi_1) \neq c_d + \rho(\phi_2) + \epsilon \hat{g}_{n}(\phi_2)$ and $c_d + \rho(\phi_1) -\epsilon \hat{g}_{n}(\phi_1) \neq c_d + \rho(\phi_2) -\epsilon \hat{g}_{n}(\phi_2)$.  
\end{assumption}
This assumption essentially ensures that there exists a $\delta$ such that the locus of the points $c_d + \rho(\phi) \pm \delta \hat{g}_{n}(\phi), \phi \in [0, 2\pi)$ form simple closed curves. Moreover, these curves do not intersect $\mathcal{C}$, if $\delta$ is measured along the unit vectors $\pm\hat{g}_{n}(\phi)$. Clearly, $c_d + \rho(\phi) -\epsilon\hat{g}_{n}(\phi) \in \text{int}(\mathcal{C})$, and $c_d + \rho(\phi) + \epsilon\hat{g}_{n}(\phi) \in \text{ext}(\mathcal{C})$ for each $\phi \in [0, 2\pi)$ with respect to Fig.~\ref{problem_fig}. In other words, Assumption~\ref{Assumption1} proposes certain requirements on $\mathcal{C}$, depending upon $\delta$. So far as the practical applications are concerned, this is a mild assumption as discussed later in the paper. 

\begin{figure}
	\centering
	\includegraphics[width=0.6\columnwidth]{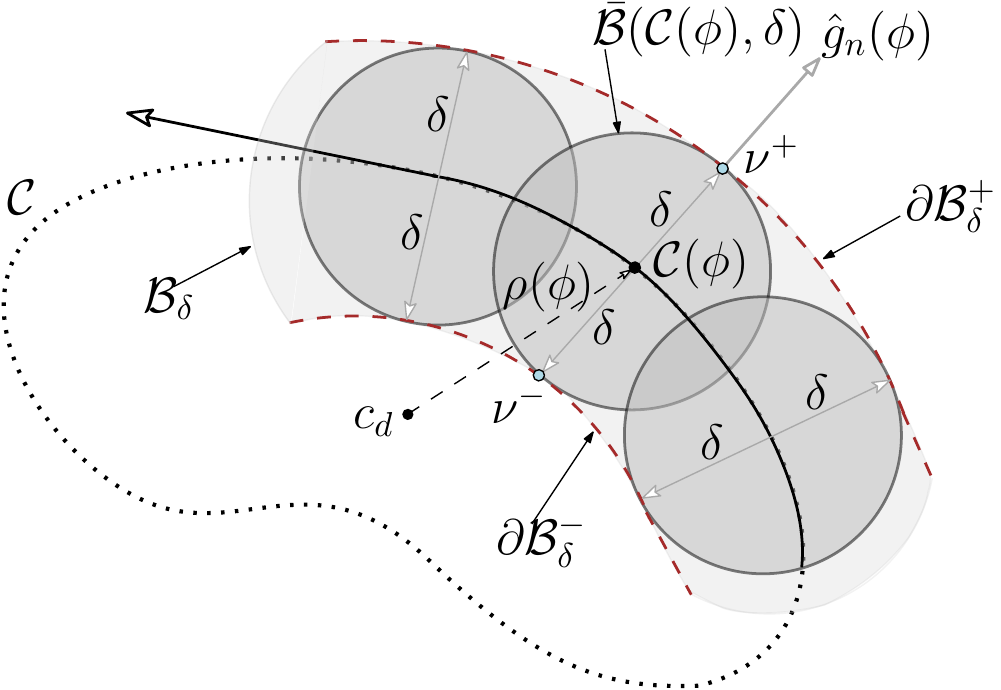}
	\caption{Illustration of the boundary $\partial \mathcal{B}$ of the set $\mathcal{B}$, mentioned in Corollary~\ref{corollary1}.}
	\label{boundary}
\end{figure}

\begin{cor}\label{corollary1}
	Under the conditions in Theorem~\ref{theorem1}, if additionally, Assumption~\ref{Assumption1} holds, then $\partial\mathcal{B}_\delta = \partial\mathcal{B}^+_{\delta} \cup \partial\mathcal{B}^-_{\delta}$, where $\partial\mathcal{B}^+_{\delta} = \bigcup_{\phi \in [0, 2\pi)}\{z \in \mathbb{C}~|~z - c_d = \rho(\phi) + \delta\hat{g}_n(\phi)\}, \partial\mathcal{B}^-_{\delta} = \bigcup_{\phi \in [0, 2\pi)}\{z \in \mathbb{C}~|~z - c_d = \rho(\phi) - \delta\hat{g}_n(\phi)\}$, and $\hat{g}_n(\phi)$ is the exterior unit normal vector to curve $\mathcal{C}$ at $\phi$, as shown in Fig.~\ref{problem_fig}.
\end{cor}

\begin{proof}
	Let $\bar{\mathcal{B}}(\mathcal{C}(\phi), \delta) \coloneqq \{z \in \mathbb{C}~|~|z - c_d - \rho(\phi)| \leq \delta\}$ be the closed disc of radius $\delta$ and center $\mathcal{C}(\phi) = c_d + \rho(\phi)$ at $\phi$. For $\nu \in \bar{\mathcal{B}}(\mathcal{C}(\phi), \delta)$, let $|\text{proj}_{\hat{g}_n(\phi)} ({\nu - \mathcal{C}(\phi)})| = |\langle\nu - \mathcal{C}(\phi), \hat{g}_n \rangle|$ be the absolute value of the projection of $\nu - \mathcal{C}(\phi)$ on $\hat{g}_n(\phi)$ at $\mathcal{C}({\phi})$. Clearly, $\max_\nu {|\text{proj}_{\hat{g}_n(\phi)} ({\nu - \mathcal{C}(\phi)})|}$ occurs when $\langle \nu - \mathcal{C}(\phi),  \hat{g}_n(\phi) \rangle = \pm\delta$, that is, along the normal vectors to $\mathcal{C}$ at $\phi$. This implies that the points $\nu^+ = \mathcal{C}(\phi) + \delta \hat{g}_n(\phi)$ and  $\nu^- = \mathcal{C}(\phi) - \delta \hat{g}_n(\phi)$ are the two farthest points in $\bar{\mathcal{B}}(\mathcal{C}(\phi), \delta)$ from the curve $\mathcal{C}$ at $\phi$, along exterior and interior normal vectors, respectively (See Fig.~\ref{boundary}). Thus, for $\phi \in [0, 2\pi)$, the locus of the points $\nu^+$ and $\nu^-$ define two boundaries $\partial\mathcal{B}^+_{\delta}$ and $\partial{B}^-_{\delta}$, respectively, as mentioned in the statement of Corollary~\ref{corollary1}. As a result,  $\partial\mathcal{B}_\delta = \partial\mathcal{B}^+_{\delta} \cup \partial\mathcal{B}^-_{\delta}$, proving our claim. 
\end{proof}

\begin{remark}
	It is to be noted that Assumption~\ref{Assumption1} is adapted to explicitly describe the boundary $\partial\mathcal{B}_\delta$ of $\mathcal{B}_\delta$. However, it is not necessary and allows any $\delta > 0$, as far as the applicability of \eqref{control1} is concerned. For instance, in the special case of a circle with radius $R > 0$, $\partial\mathcal{B}_\delta = \{z \in \mathbb{C}~|~|z - c_d| = R + \delta\}$ if $\delta > R$, and $\partial\mathcal{B}_\delta = \{z \in \mathbb{C}~|~|z - c_d| = R - \delta\} \cup \{z \in \mathbb{C}~|~|z - c_d| = R + \delta\}$ if $\delta < R$ \cite{jain2019trajectory}. 
\end{remark}


Additionally, the following theorem relates perimeters and areas of the regions enclosed by $\partial\mathcal{B}_{\delta}$ and curve $\mathcal{C}$, under Assumption~\ref{Assumption1}. 


\begin{thm}\label{thm_perimeter_area}
	Let $\Gamma_{\mathcal{C}}, {\Gamma}_{\partial\mathcal{B}_\delta}$ be the respective perimeters of $\mathcal{C}$ and $\partial\mathcal{B}_{\delta}$, and $\mathcal{A}_{\mathcal{C}}, {\mathcal{A}}_{\partial\mathcal{B}_\delta}$ the areas enclosed by them in Corollary~\ref{corollary1}. Then, it holds that ${\Gamma}_{\partial\mathcal{B}_\delta} = 2\Gamma_{\mathcal{C}}, {\mathcal{A}}_{\partial\mathcal{B}_\delta} = 2\delta\Gamma_{\mathcal{C}}$. 
\end{thm}

Before proving Theorem~\ref{thm_perimeter_area}, we state the following preliminary definition and result from \cite{tapp2016differential,pressley2010elementary}. 

\begin{defn}[\hspace{-.1pt}Orientation of a Planar Curve \cite{tapp2016differential}, pg.~62]
	A simple plane closed curve $\alpha : [a, b] \to \mathbb{C}$ is called positively-oriented if, for each $t \in [a, b]$, $i{\rm e}^{\arg(\dot{\alpha})}$ points into $\text{int}(\alpha)$ in the sense that there exists $\tilde{\delta}$ such that $\alpha(t) + i s{\rm e}^{\arg(\dot{\alpha})}$ lies in $\text{int}(\alpha)$ for all $s \in (0, \tilde{\delta})$. Otherwise, $\alpha$ is negatively-oriented, in which case $i{\rm e}^{\arg(\dot{\alpha})}$ points towards $\text{ext}(\alpha)$ for all $t \in [a, b]$. 
\end{defn}

\begin{thm}[\hspace{-.1pt}Hopf's Umlaufsatz \cite{pressley2010elementary}, pg.~57]\label{thm_Holf_Umlauf}
	Let $\alpha : [a, b] \to \mathbb{R}^2$ be a simple closed plane curve with curvature function $\kappa$. Then, its total signed curvature, $\int_{a}^{b} \kappa(t)dt = \pm 2\pi I$, where $I$ is referred to as rotation index, and is equal to $+1 (\text{resp.}, -1)$ for positively-oriented (resp., negatively-oriented) curves.
\end{thm}

We are now ready to prove Theorem~\ref{thm_perimeter_area}.

\begin{proof}
	Analogous to the notations in Theorem~\ref{thm_perimeter_area}, denote by ${\Gamma}_{\partial\mathcal{B}^+_\delta}$ and ${\Gamma}_{\partial\mathcal{B}^-_\delta}$, the perimeters of the boundaries $\partial\mathcal{B}^+_{\delta}$ and  $\partial\mathcal{B}^-_{\delta}$, defined in Corollary~\ref{corollary1}, and by $\mathcal{A}_{\partial\mathcal{B}^+_\delta}$ and ${\mathcal{A}}_{\partial\mathcal{B}^-_\delta}$, the areas enclosed by them, respectively. Let $d\mu$ be the argument of tangential vector to a differential arc-length $d\Gamma_{\mathcal{C}}$ of the curve $\mathcal{C}$. Considering two adjacent normals to $\mathcal{C}$ at the end points of $d\Gamma_{\mathcal{C}}$ and neglecting higher order terms, one can write $
	d{\Gamma}_{\partial\mathcal{B}^+_\delta} = d\Gamma_{\mathcal{C}} + \delta d\mu; d{\Gamma}_{\partial\mathcal{B}^-_\delta} = d\Gamma_{\mathcal{C}} - \delta d\mu$, and $d({\mathcal{A}}_{\partial\mathcal{B}^+_\delta} - \mathcal{A}_{\mathcal{C}}) = \frac{1}{2}\delta(d\Gamma_{\mathcal{C}} + d{\Gamma}_{\partial\mathcal{B}^+_\delta}) = \frac{1}{2}\delta(2d\Gamma_{\mathcal{C}} + \delta d\mu); d(\mathcal{A}_{\mathcal{C}} - {\mathcal{A}}_{\partial\mathcal{B}^-_\delta}) = \frac{1}{2}\delta(d\Gamma_{\mathcal{C}} + d{\Gamma}_{\partial\mathcal{B}^-_\delta}) =  \frac{1}{2}\delta(2d\Gamma_{\mathcal{C}} - \delta d\mu)$.
	The term $d\mu$ can be written in terms of $d\phi$ as $d\mu = \kappa(\phi)|d\rho/d\phi|d\phi$ (Subsection~\ref{curve_def}). The quantity $\kappa_T = \int_{0}^{2\pi} \kappa(\phi)|d\rho/d\phi|d\phi$ is unchanged by re-parametrization [\cite{tapp2016differential}, pg. 78]. Thus, using arc-length re-parameterization, we have that $\kappa_T = \int_{0}^{\Gamma_{\mathcal{C}}} \kappa(\sigma)|d\rho/d\sigma|d\sigma = \int_{0}^{\Gamma_{\mathcal{C}}} \kappa(\sigma)d\sigma = 2\pi I$, using Hopf's Umlaufsatz in Theorem~\ref{thm_Holf_Umlauf}, and $|d\rho/d\sigma| = 1$, according to \eqref{arc_length}. Under Assumption~\ref{Assumption1}, one can note that $\mathcal{C}$, $\partial\mathcal{B}^+_\delta$ and $\partial\mathcal{B}^-_\delta$ are \emph{positively-oriented} simple closed plane curves with respect to Fig.~\ref{problem_fig}, and hence, rotation index $I = +1$. Using this fact, while integrating previous expressions for a complete circuit, yields ${\Gamma}_{\partial\mathcal{B}^+_\delta}  = \Gamma_{\mathcal{C}} + 2\pi\delta; {\Gamma}_{\partial\mathcal{B}^-_\delta}  = \Gamma_{\mathcal{C}} - 2\pi\delta$, and ${\mathcal{A}}_{\partial\mathcal{B}^+_\delta} = \mathcal{A}_{\mathcal{C}} + \delta\Gamma_{\mathcal{C}} + \pi\delta^2;{\mathcal{A}}_{\partial\mathcal{B}^-_\delta} = \mathcal{A}_{\mathcal{C}} - \delta\Gamma_{\mathcal{C}} + \pi\delta^2 \implies {\Gamma}_{\partial\mathcal{B}_{\delta}} = {\Gamma}_{\partial\mathcal{B}^+_{\delta}} + {\Gamma}_{\partial\mathcal{B}^-_{\delta}} = 2\Gamma_{\mathcal{C}}$, and ${\mathcal{A}}_{\partial\mathcal{B}_{\delta}} = {\mathcal{A}}_{\partial\mathcal{B}^+_{\delta}} - {\mathcal{A}}_{\partial\mathcal{B}^-_{\delta}} = 2\delta\Gamma_{\mathcal{C}}$, as claimed. 
\end{proof}

\begin{remark}
	The relations in Theorem~\ref{thm_perimeter_area} can also be written in terms of area $\mathcal{A}_{\mathcal{C}}$ (resp., global maximum curvature $\kappa_{\text{max}} = \max_{\phi}|\kappa(\phi)|$) of curve $\mathcal{C}$ using the isoperimetric inequality ${\Gamma}^2_{\mathcal{C}} \geq 4\pi\mathcal{A}_{\mathcal{C}}$ (resp., ${\Gamma}_{\mathcal{C}} \geq 2\pi/\kappa_{\text{max}}$) for a simple closed plane curve $\mathcal{C}$, where equality holds if and only if $\mathcal{C}$ is a circle [\cite{tapp2016differential}, pg. 81, $\&$ 98].
\end{remark}

In several practical applications, we would often like to restrict the motion of the vehicles in a workspace within the outer boundary. For any convex curve $\mathcal{C}$, the exterior normals never intersect irrespective of any $\delta$, and hence, the ideas in Theorem~\ref{thm_perimeter_area} are applicable if one is interested to know the perimeter and area enclosed by the outer boundary.  

\section{Bounds on Various Signals}\label{section5}
This section obtains bounds on various intermediate signals based on Theorem~\ref{theorem1}. We begin by stating the following theorem:

\begin{thm}[Curve-Phase Synchronization]\label{bounds_synchronization}
	Let $\mathcal{L}$ be the Laplacian of an undirected and connected graph $\mathcal{G}$ with $N$ vertices. Consider the closed-loop system \eqref{modelNew}, under the saturated control law \eqref{saturated_control} with $\zeta_k$ given by \eqref{control1}, where $K_{\mathcal{C}} > 0$, and $K < 0$ for all $k$. Assume that the initial states of the agents belong to the set $\mathcal{Z}_\delta$, as defined in Theorem~\ref{theorem1}. Then, the following properties hold: 
	\begin{enumerate}
	\item[i)] The absolute value of $e_k$, and $r_k$, for all $k$, are bounded by
	\begin{equation*}
	|e_k| = |r_k - c_d - \rho(\phi)|  \leq \delta \sqrt{1-{\rm e}^{-\left(\frac{2\mathcal{V}_1(0)}{K_{\mathcal{C}}}\right)}}.
	\end{equation*}
	\item[ii)] The squared summation of the absolute value of the relative curve-phasors ${\rm e}^{i\psi_j} - {\rm e}^{i\psi_k}$ belongs to the compact set
	\begin{equation*}
	\sum_{\{j,k\} \in \mathcal{E}}|{\rm e}^{i\psi_j} - {\rm e}^{i\psi_k}|^2 \in \left[0, \min\left\{-\left(\frac{4\pi \mathcal{V}_1(0)}{K\Gamma_{\mathcal{C}}}\right), 4|\mathcal{E}|\right\}\right],
	\end{equation*}
	where, $\mathcal{V}_1(0) = \mathcal{V}_1(\pmb{r}(0), \pmb{\theta}(0))$, and $|\mathcal{E}|$ is the cardinality of the edge set $\mathcal{E}$ of the graph $\mathcal{G}$, respectively.
	\end{enumerate}
\end{thm} 

\begin{proof}
	\begin{enumerate}
	\item[i)] Following Theorem~\ref{theorem1}, since $\mathcal{V}_1(\pmb{r}(t), \pmb{\theta}(t)) \leq \mathcal{V}_1(0), \forall t \geq 0$, in $\mathcal{Z}_\delta$, it follows from \eqref{V_1} that $\frac{K_{\mathcal{C}}}{2} \sum_{k=1}^{N} \ln \left(\frac{\delta^2}{\delta^2 - |e_k(t)|^2}\right) \leq \mathcal{V}_1(0) \implies  \ln \left(\frac{\delta^2}{\delta^2 - |e_k(t)|^2}\right) \leq \frac{2\mathcal{V}_1(0)}{K_{\mathcal{C}}}$, $\forall k$, and $\forall t \geq 0$, in $\mathcal{Z}_\delta$. Taking exponential on both side, yields ${\delta^2}/({\delta^2 - |e_k(t)|^2}) \leq {\rm e}^{\frac{2\mathcal{V}_1(0)}{K_{\mathcal{C}}}}$. It has been established in Theorem~\ref{theorem1} that $|e_k(t)|<\delta, \forall k$, and $\forall t$, implying that $\delta^2 - |e_k(t)|^2 > 0, \forall k$, and $\forall t$. Thus, we obtain $\delta^2 \leq {\rm e}^{\frac{2\mathcal{V}_1(0)}{K_{\mathcal{C}}}}({\delta^2 - |e_k(t)|^2}) \implies |e_k(t)| \leq \delta \sqrt{1-{\rm e}^{-({2\mathcal{V}_1(0)}/{K_{\mathcal{C}}})}}, \forall k$ and $\forall t \geq 0$. Further, substituting for $e_k$ from \eqref{error_var}, it implies that $|r_k(t) - c_d - \rho(\phi(t))| \leq \delta \sqrt{1-{\rm e}^{-({2\mathcal{V}_1(0)}/{K_{\mathcal{C}}})}}, \forall k$ and $\forall t \geq 0$. 
	
	\item[ii)] A similar argument as above for the second term on RHS in \eqref{V_1}, results in $\mathcal{W}(\pmb{\psi}) \leq -\frac{2\pi\mathcal{V}_1(0)}{K\Gamma_{\mathcal{C}}}$. Substituting for $\mathcal{W}(\pmb{\psi})$ from \eqref{W}, one can write that $\langle {\rm e}^{i\pmb{\psi}}, \mathcal{L}{\rm e}^{i\pmb{\psi}}\rangle \leq - \frac{4\pi  \mathcal{V}_1(0)}{K\Gamma_{\mathcal{C}}}$. However, it follows from \eqref{W_max} that $\langle {\rm e}^{i\pmb{\psi}}, \mathcal{L}{\rm e}^{i\pmb{\psi}}\rangle = \sum_{\{j,k\} \in \mathcal{E}}|{\rm e}^{i\psi_j} - {\rm e}^{i\psi_k}|^2 \leq 2^2|\mathcal{E}|$, for any undirected and connected graph $\mathcal{G}$. Thus, the required bounds on $\sum_{\{j,k\} \in \mathcal{E}}|{\rm e}^{i\psi_j} - {\rm e}^{i\psi_k}|^2$, as given in the theorem, is obtained.
	\end{enumerate} 
\end{proof}

\begin{thm}[Curve-Phase Balancing]\label{bounds_balancing}
	Let $\mathcal{L}$ be the Laplacian of an undirected and connected circulant graph $\mathcal{G}$ with $N$ vertices. Consider the closed-loop system \eqref{modelNew}, under the saturated control law \eqref{saturated_control} with $\zeta_k$ given by \eqref{control1}, where $K_{\mathcal{C}} > 0$, and $K > 0$ for all $k$. Assume that the initial states of the agents belong to the set $\mathcal{Z}_\delta$, as defined in Theorem~\ref{theorem1}. Then, the following properties hold: 
	\begin{enumerate}
	\item[i)] The absolute value of $e_k$, and $r_k$, for all $k$, are bounded by 
	\begin{equation*}
	|e_k| = |r_k - c_d - \rho(\phi)| \leq \delta \sqrt{1-{\rm e}^{-\left(\frac{2\mathcal{V}_2(0)}{K_{\mathcal{C}}}\right)}}.	
	\end{equation*}
	\item[ii)] The squared summation of the absolute value of ${\rm e}^{i\psi_j} - {\rm e}^{i\psi_k}$ belongs to the compact set
	\begin{equation*}
	\sum_{\{j,k\} \in \mathcal{E}}|{\rm e}^{i\psi_j} - {\rm e}^{i\psi_k}|^2  \in \mathcal{J},
	\end{equation*}
	where, 
	\begin{equation*}
	\mathcal{J} =  \left[\max\left\{0, \left(N\lambda_{\text{max}} - \left(\frac{4\pi \mathcal{V}_2(0)}{K\Gamma_{\mathcal{C}}}\right)\right)\right\}, N\lambda_{\text{max}}\right],
	\end{equation*}
	and $\mathcal{V}_2(0) = \mathcal{V}_2(\pmb{r}(0), \pmb{\theta}(0))$, and $|\mathcal{E}|$, is as defined in Theorem~\ref{bounds_synchronization}.
	\end{enumerate} 
\end{thm} 

\begin{proof}
	For any undirected and connected \emph{circulant} graph $(N/2)\lambda_{\text{max}} \leq 2|\mathcal{E}|$, with equality if and only if the circulant graph $\mathcal{G}$ forms a ring topology (that is, the minimally connected circulant graph). Thus, the bound in case ii) of Theorem~\ref{bounds_balancing} is different than that from Theorem~\ref{bounds_synchronization}. The rest of the proof follows along the similar steps as in Theorem~\ref{bounds_synchronization}, and hence omitted. 
\end{proof}


\begin{remark}
	Note that Corollary~\ref{corollary1} and Theorem~\ref{thm_perimeter_area} can be written equivalently, individually for curve-phase synchronization and balancing, by replacing $\delta$ with ${\delta}_s = \delta \sqrt{1-{\rm e}^{-({2\mathcal{V}_1(0)}/{K_{\mathcal{C}}})}} < \delta$ and ${\delta}_b = \delta \sqrt{1-{\rm e}^{-({2\mathcal{V}_2(0)}/{K_{\mathcal{C}}})}} < \delta$, derived using the tighter bound $|r_k - c_d - \rho(\phi)| \leq \delta \sqrt{1-{\rm e}^{-({2\mathcal{V}_1(0)}/{K_{\mathcal{C}}})}}$ and $|r_k - c_d - \rho(\phi)| \leq \delta \sqrt{1-{\rm e}^{-({2\mathcal{V}_2(0)}/{K_{\mathcal{C}}})}}$, on the trajectories of the agents in Theorem~\ref{bounds_synchronization} and Theorem~\ref{bounds_balancing}, respectively.  
\end{remark}

The prerequisite of our approach is that the agents' initial conditions must satisfy $|e_k(0)| < \delta$ for all $k$. We characterize the feasible initial conditions in the following theorem.  

\begin{thm}\label{thm_initial_conditions}
	The condition $|e_k(0)| < \delta$ is satisfied for all $k$, if the initial conditions (that is, initial positions $r_k(0)$ and heading angles $\theta_k(0)$) of the agents in \eqref{modelNew} belong to the set $\mathcal{B}_\delta$, where, $e_k$ and $\mathcal{B}_\delta$ are defined in \eqref{error_var} and Theorem~1, respectively.  
\end{thm} 

\begin{proof}
	The proof directly follows from Theorem~1 and Corollary~1. Note that there exists at least one setting of the initial conditions such that $|e_k(0)| < \delta$ is satisfied for all $k$, as $\theta_k(0) \in \mathbb{S}^1$ are free states. 
\end{proof}

\begin{figure}[t!]
	\centering
	\subfigure[Topology]{\includegraphics[width=2.5cm]{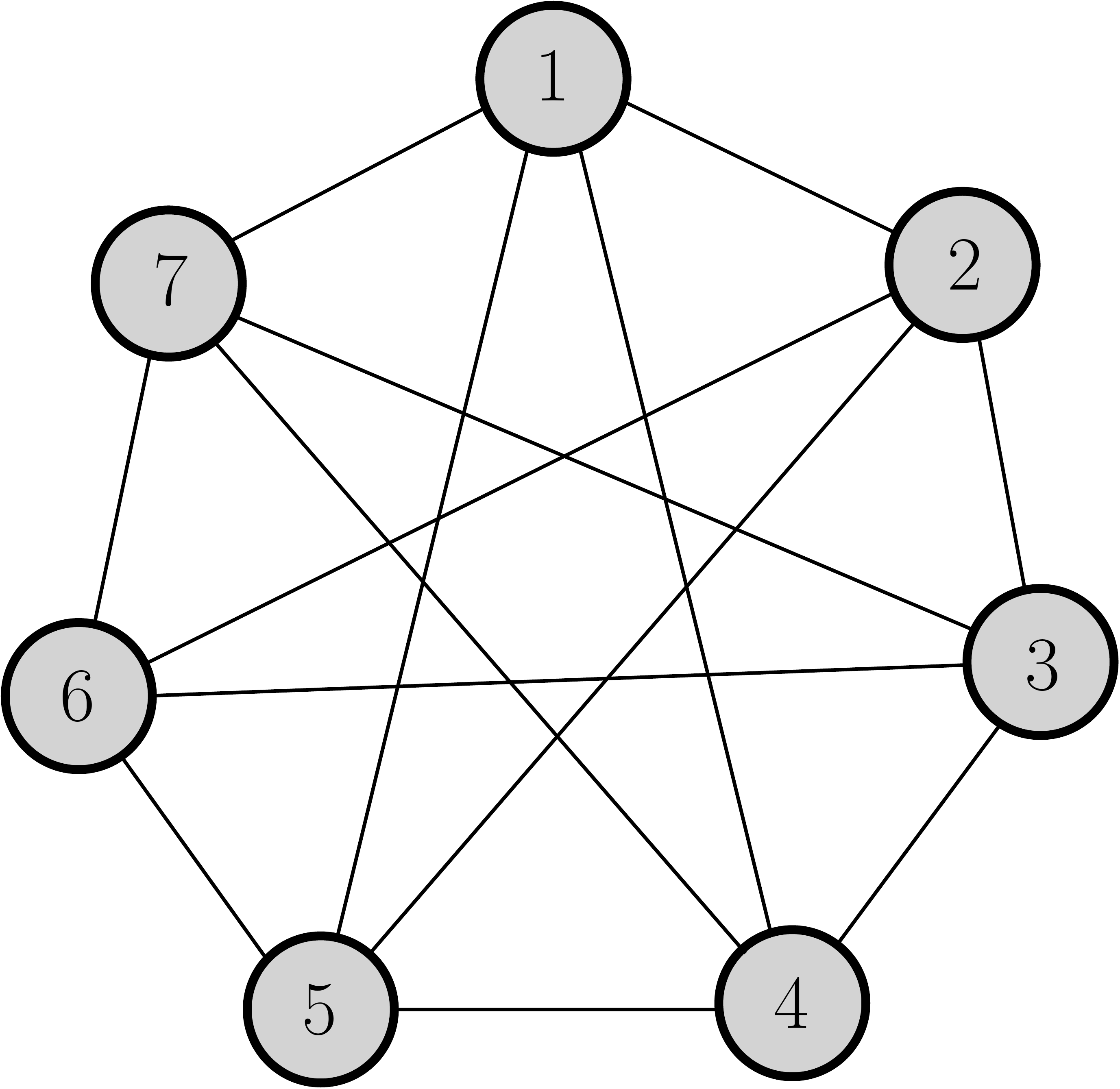}}\hspace{1.0cm}
	\subfigure[Laplacian $\mathcal{L}$]{\includegraphics[width=2.5cm]{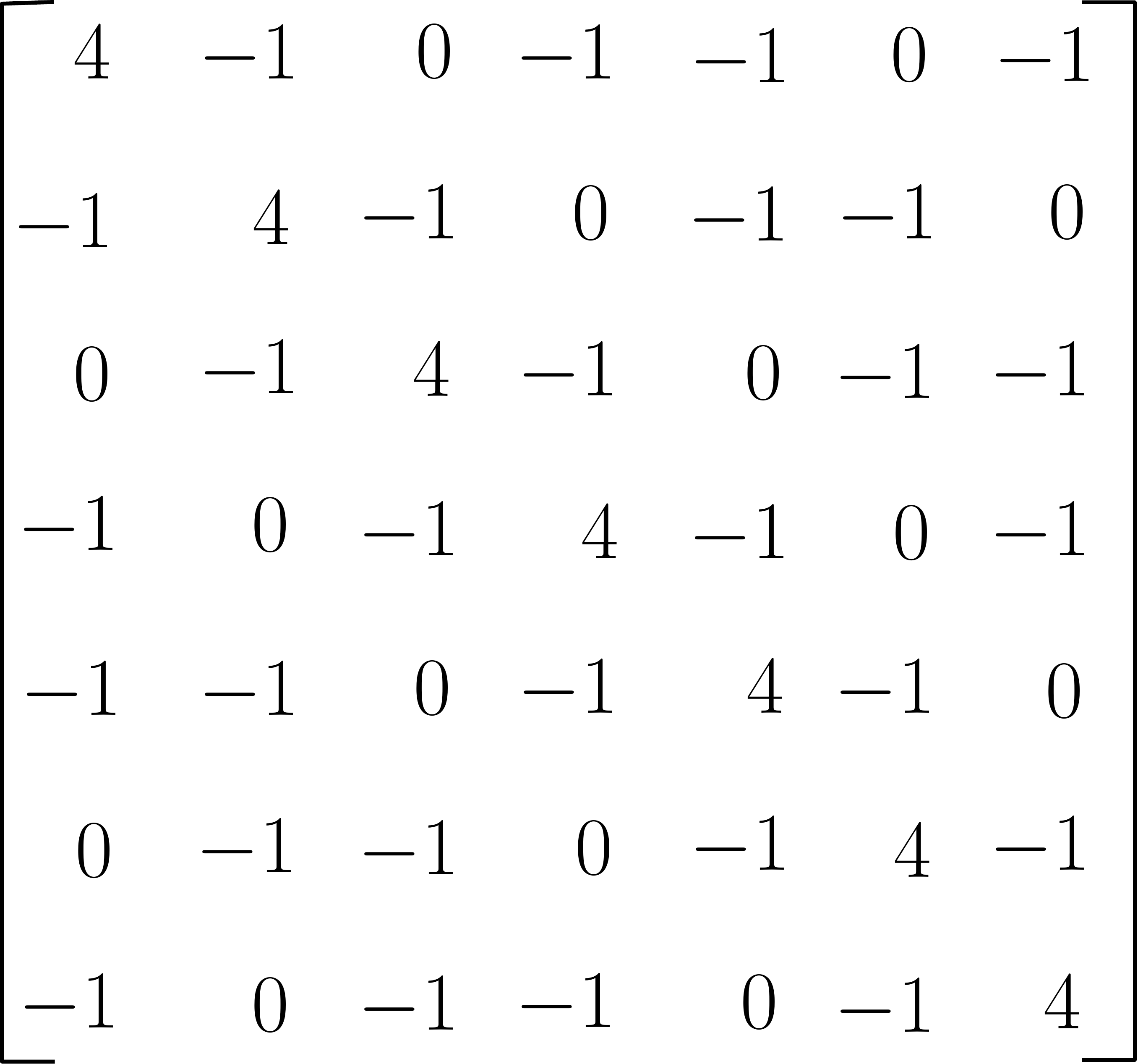}}
	\caption{Interaction topology and associated Laplacian for $N=7$ agents.}
	\label{comm_top}
\end{figure}

\begin{figure}[t!]
	\centering
	\subfigure[Synchronization]{\includegraphics[width=3.5cm]{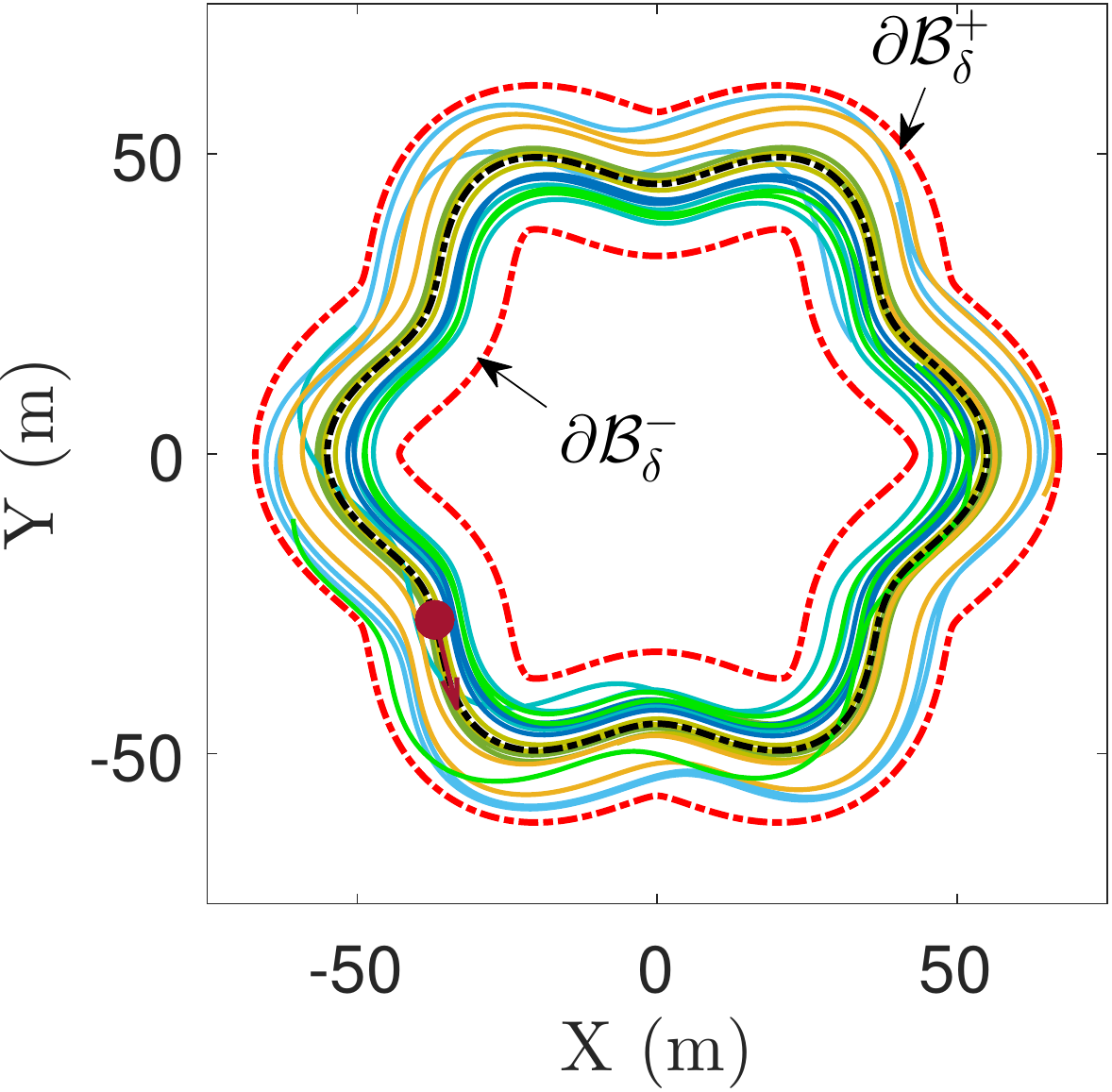}}\hspace{0.2cm}
	\subfigure[Balancing]{\includegraphics[width=3.5cm]{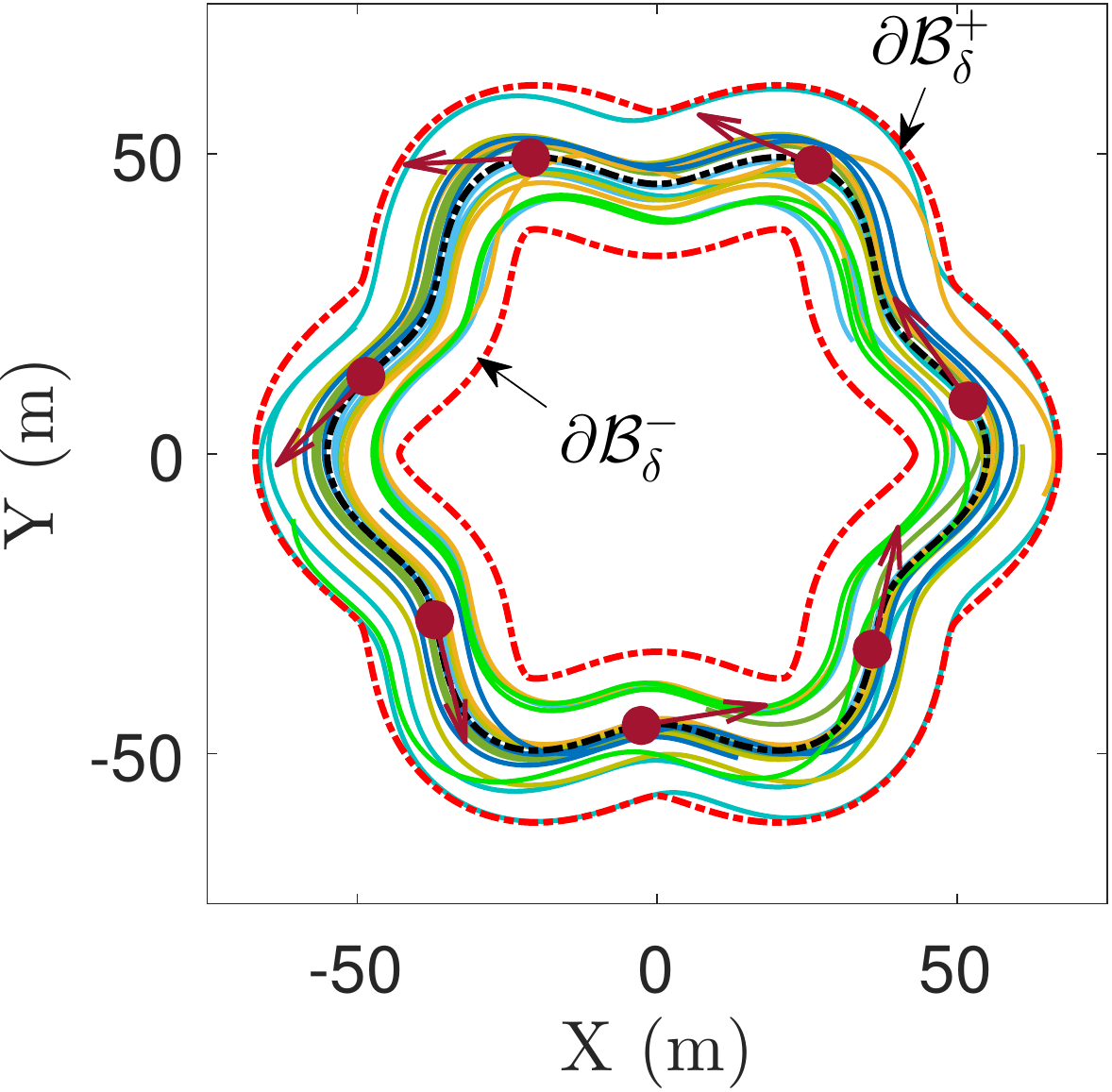}}\\
	\subfigure[Errors$-$synchronization]{\includegraphics[width=3.5cm]{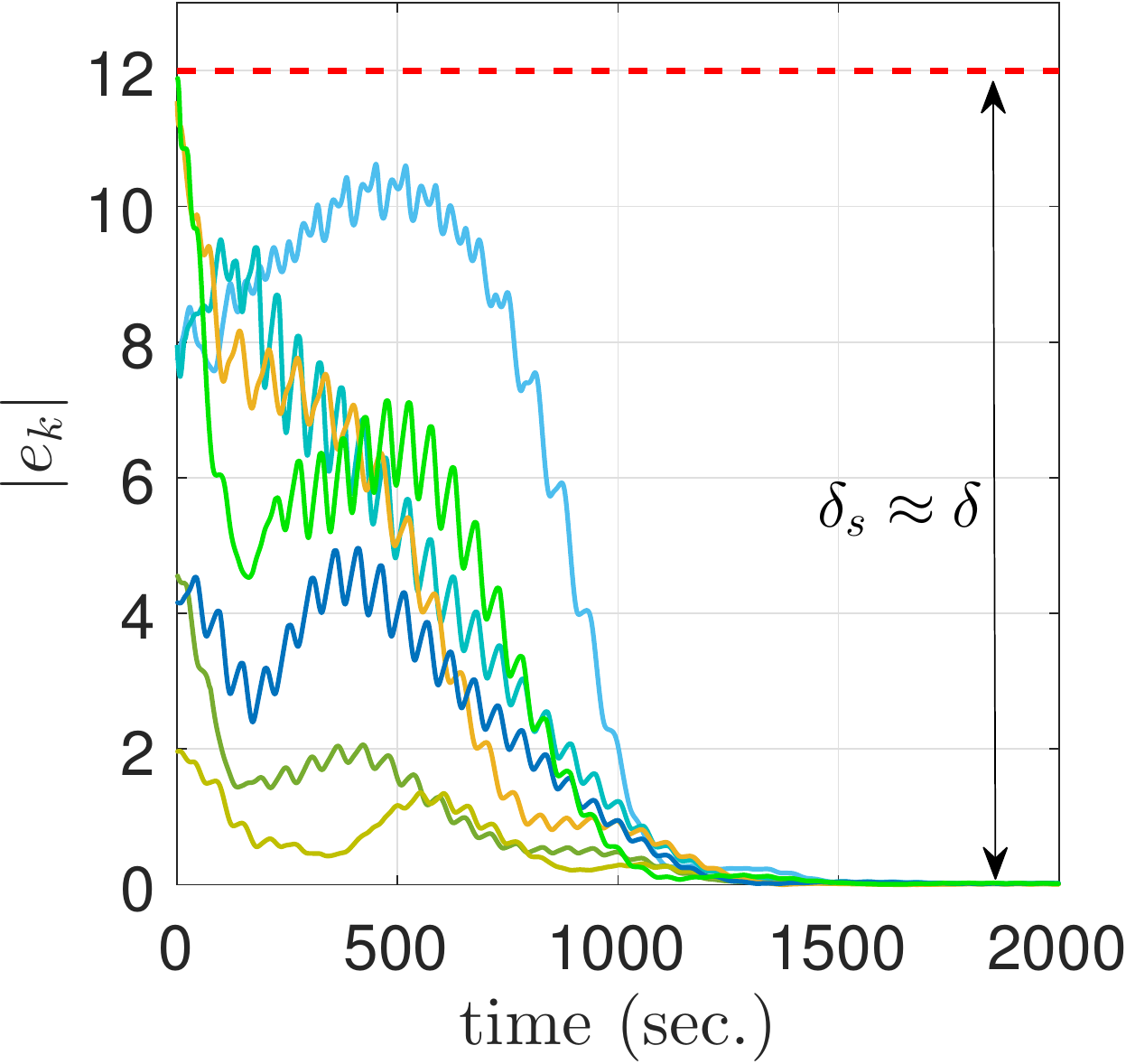}}\hspace{0.5cm}
	\subfigure[Errors$-$balancing]{\includegraphics[width=3.5cm]{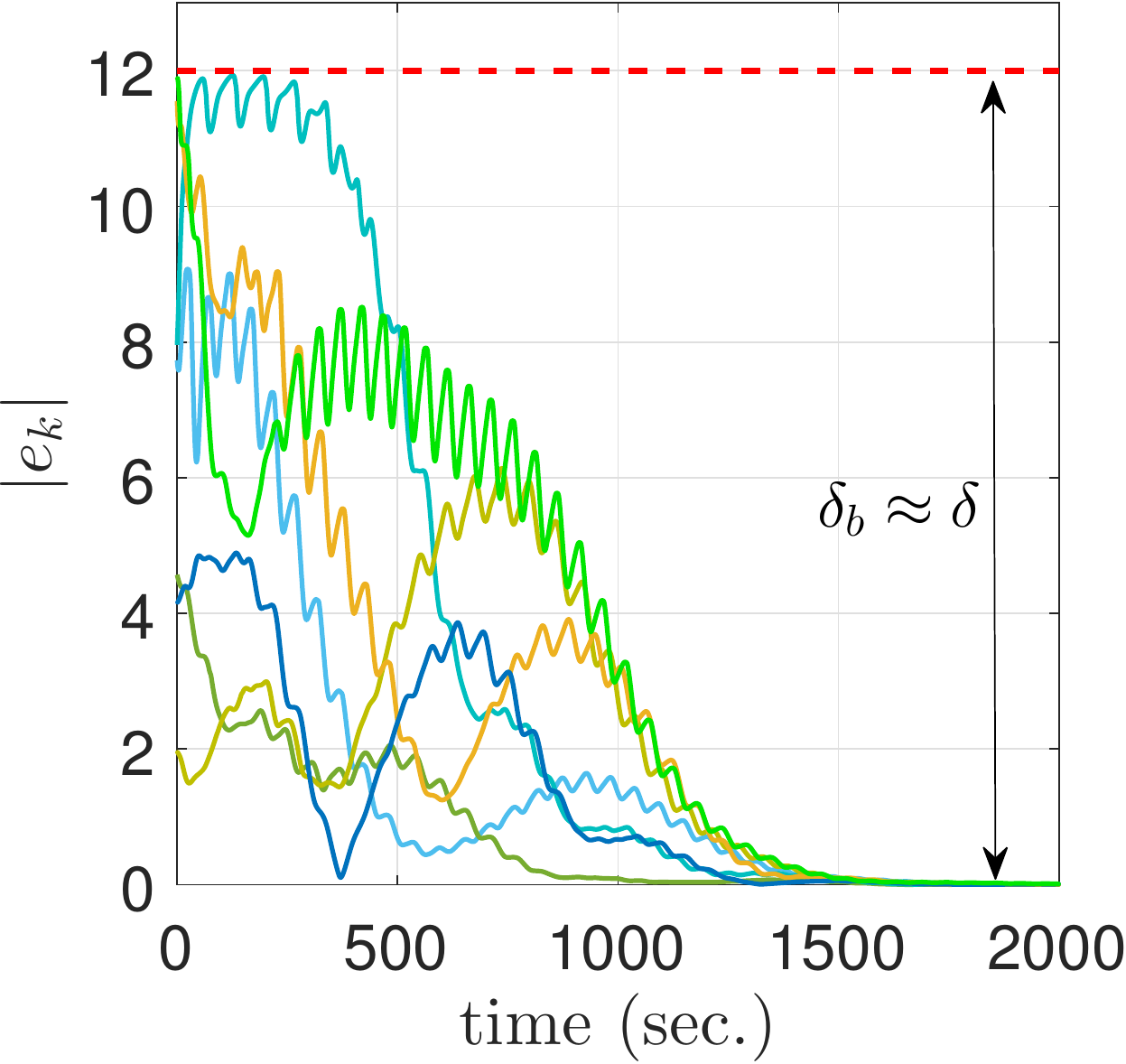}}
	\caption{Agents' trajectories and the absolute errors $|e_k|$ with time.}
	\label{traj_err_plots}
\end{figure}

\section{Simulation Results}\label{section6}
Consider seven agents ($N=7$) with an interaction topology given by a circulant graph $\mathcal{G}$ in Fig.~\ref{comm_top}.
Let us stabilize the agents around the polar-rose curve, as discussed in Example~\ref{example1}, with parameters $\tilde{a} = 10,\tilde{b} = 6, \tilde{s} = 5$, and center at $c_d = (0,0)$. Assume $\delta = 12$. The initial positions and heading angles of the agents are randomly chosen to satisfy $|e_k(0)| < \delta$ for all $k = 1, \ldots, 7$, according to Theorem~\ref{thm_initial_conditions}, and are $\pmb{x}(0)  = [32.6, 8.1, -50.2, -6.7, 64.4, -46.1, -60.6]^T$, $\pmb{y}(0)  = [18.7, -42.5, 21.2, -48.2, -7.0, -9.2, -10.8]^T$, and $\pmb{\theta}(0) = [127.3^{\circ}, 341.2^{\circ}, 222.6^{\circ}, 18.5^{\circ}, 59.5^{\circ}, 314.3^{\circ}, 271.7^{\circ}]^T$. Since $\delta < \text{min}_{\phi}|1/\kappa(\phi)| = 12.87$, one can easily observe that Assumption~1 holds for this curve, and hence, there exist inner and outer boundaries $\partial\mathcal{B}_{\delta}^-$ and $\partial\mathcal{B}_{\delta}^+$, as defined in Corollary~\ref{corollary1}. 
\begin{itemize}
\item Fig.~\ref{traj_err_plots} shows the agents' trajectories and errors $e_k$ for both curve-phase synchronization and balancing. The results are obtained under control law \eqref{control1} with gains $K_{\mathcal{C}} = 2.5$ and $K = -0.1~(\text{resp.}, 0.2)$ for synchronization (resp., balancing). It is clearly seen that the agents achieve synchronization and balancing, and their trajectories stay within the set $\mathcal{B}_{\delta}$, bounded by $\partial\mathcal{B}_{\delta}^-$ and $\partial\mathcal{B}_{\delta}^+$. Moreover, the absolute value of errors $|e_k|$ for all $k = 1, \ldots, 7$, are bounded by $\delta_s \approx \delta$ (resp., $\delta_b \approx \delta$) for synchronization (resp., balancing) and approaches zero, as desired. From an application point of view, one can consider that the agents are moving in different planes in curve-phase synchronization \cite{leonard2007collective,jain2019trajectory}. 

\begin{figure}[t!]
	\centering
	\subfigure[Control$-$synchronization]{\includegraphics[width=3.5cm]{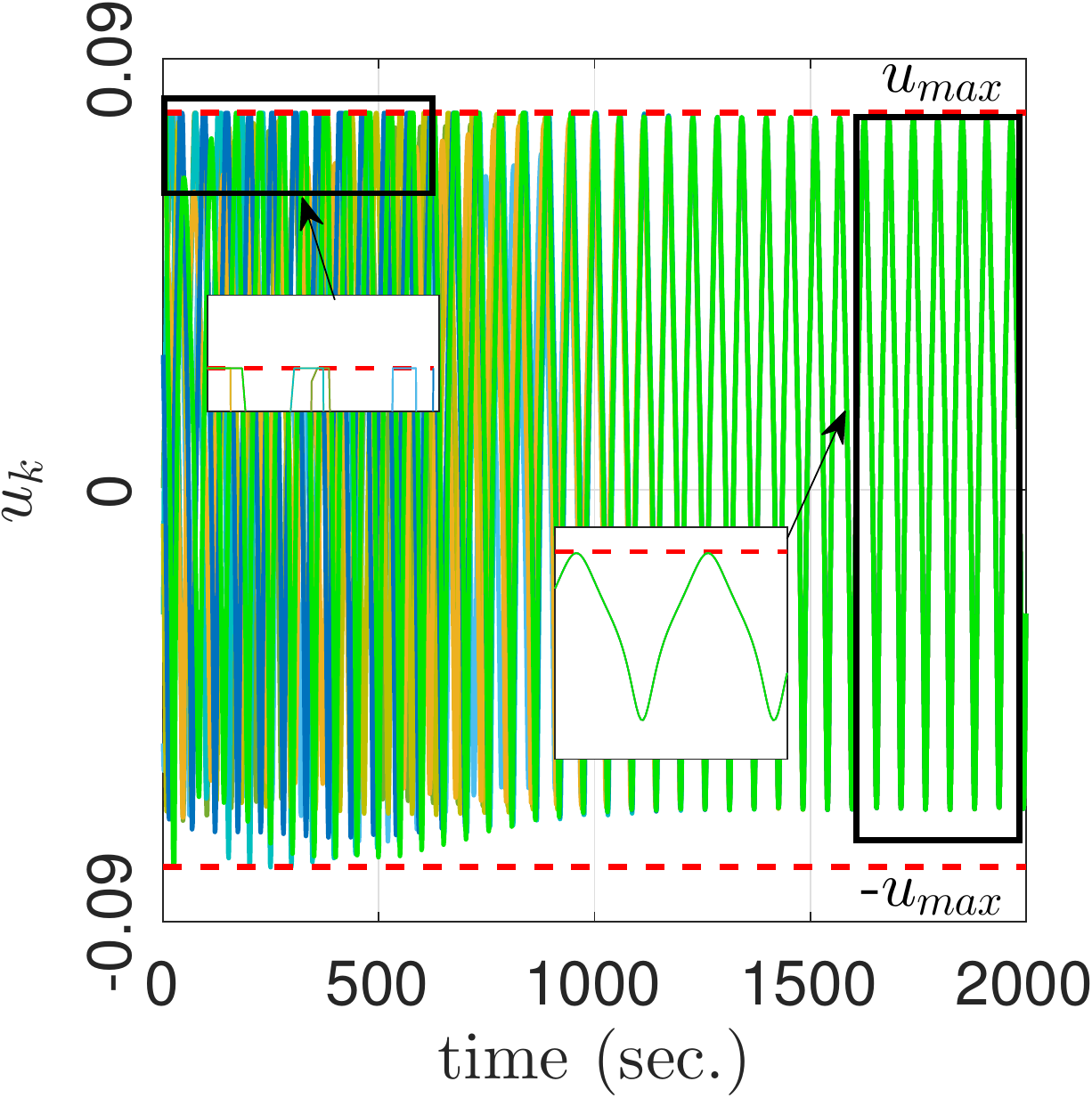}}\hspace{0.5cm}
	\subfigure[Control$-$balancing]{\includegraphics[width=3.5cm]{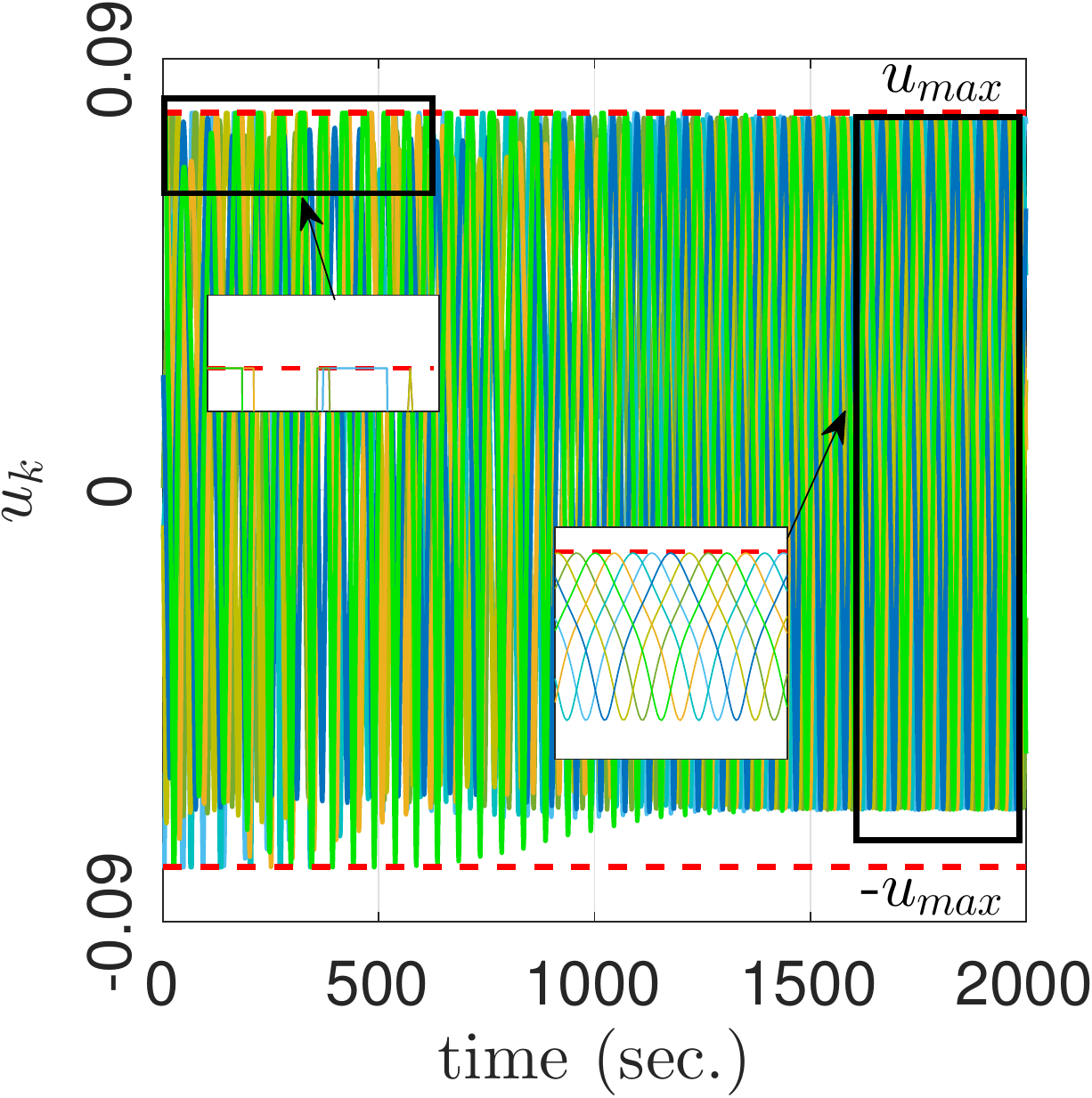}}
	\subfigure[$\zeta_k-$synchronization]{\includegraphics[width=3.5cm]{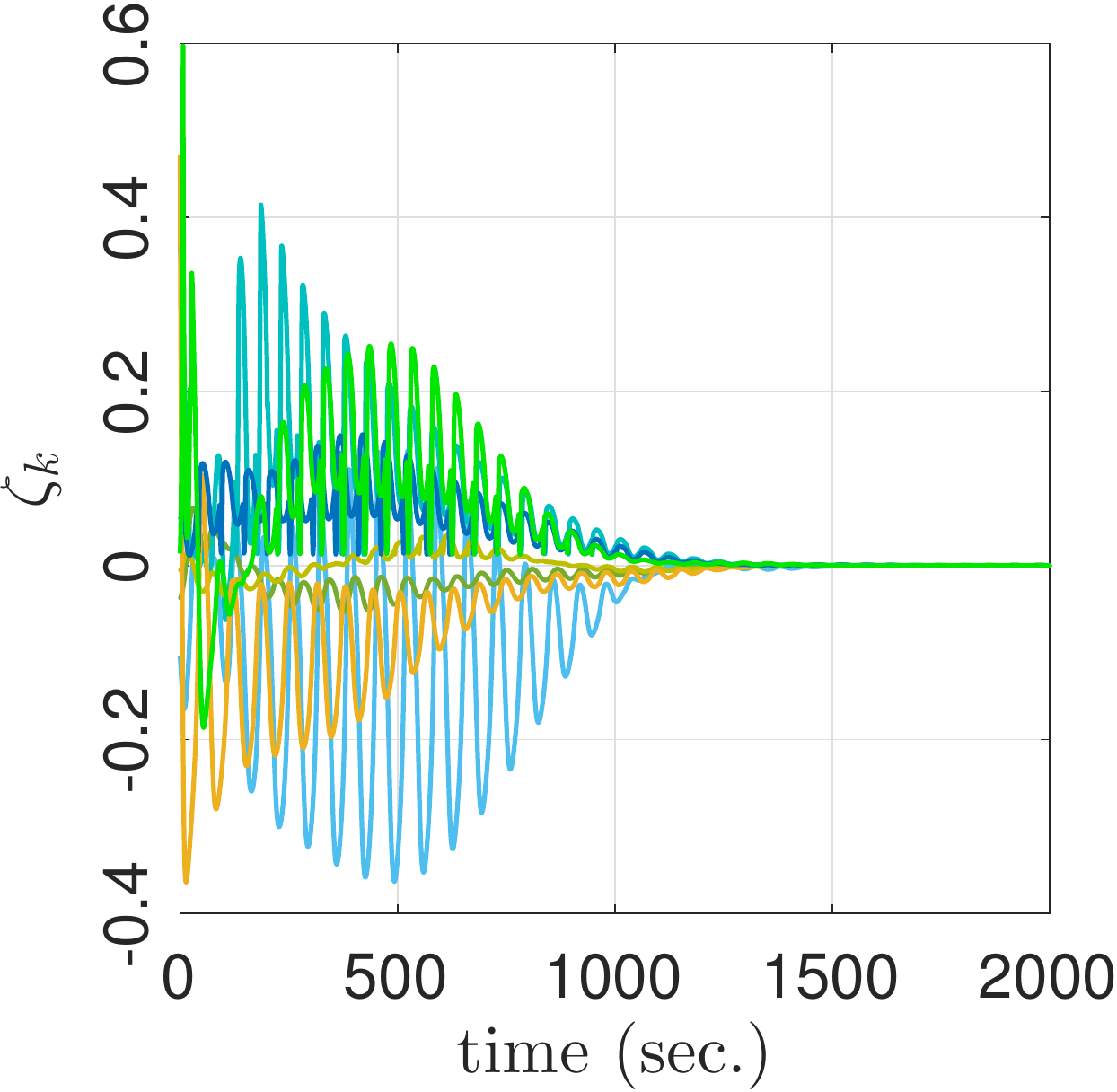}}\hspace{0.5cm}
	\subfigure[$\zeta_k-$balancing]{\includegraphics[width=3.5cm]{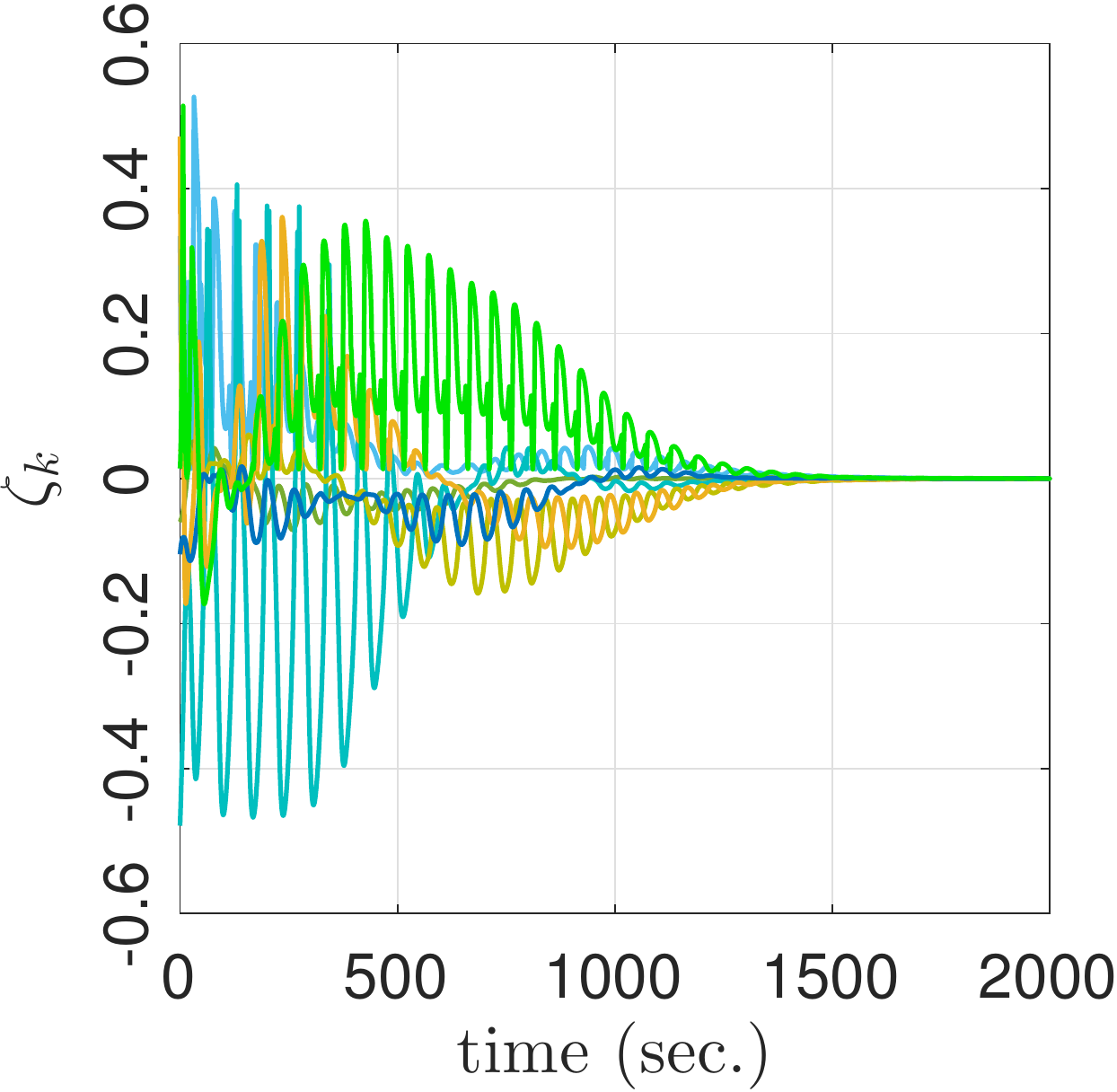}}	
	\caption{Control inputs $u_k$ in \eqref{saturated_control} with time.}
	\label{ctrl_plots}
\end{figure}

\begin{figure}
	\centering
	\subfigure[$\hspace*{-0.08cm}|p_{\psi}|,\mathcal{W}_{\psi}-$synchronization]{\includegraphics[width=3.5cm]{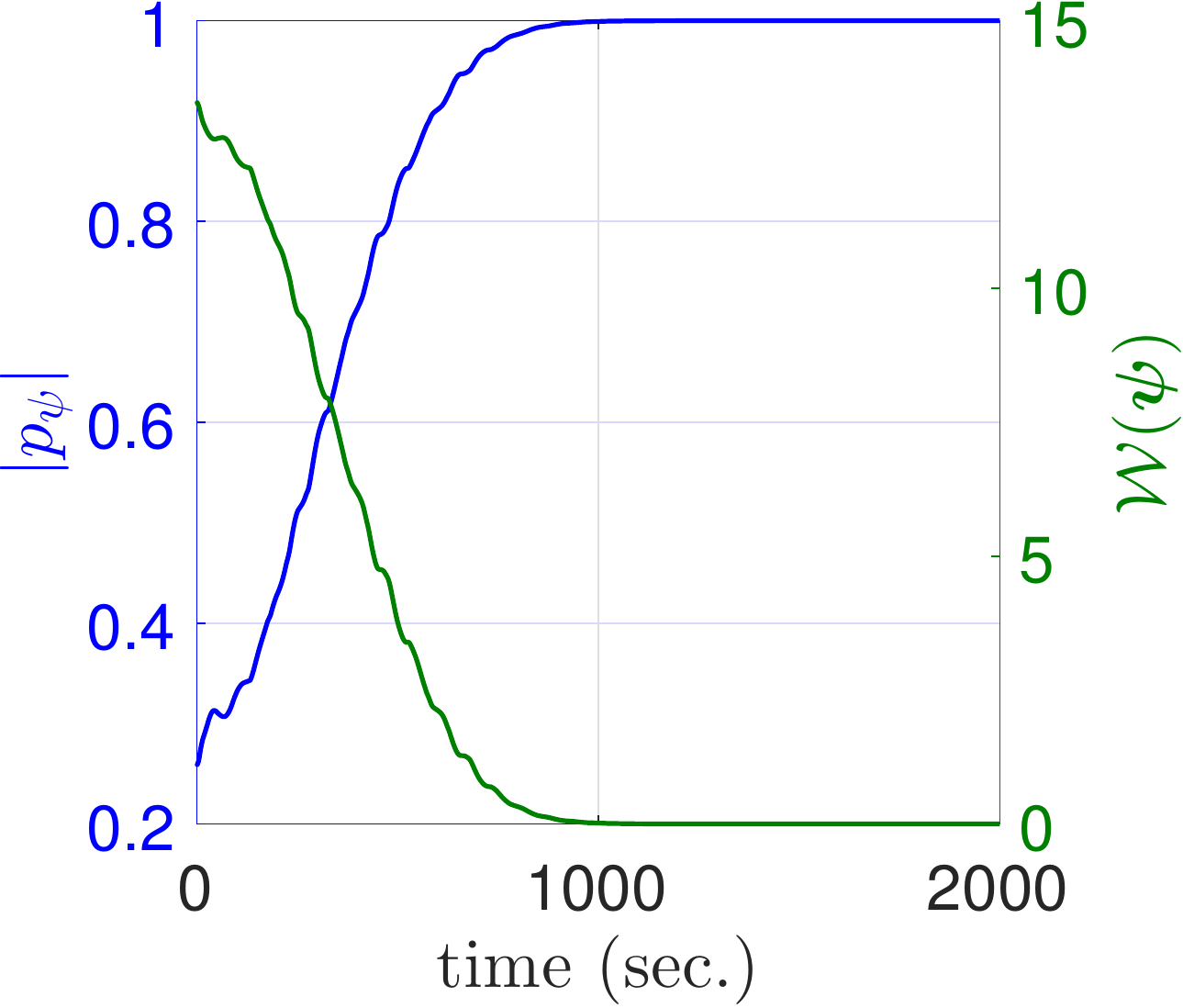}}\hspace{0.5cm}
	\subfigure[$|p_{\psi}|,\mathcal{W}_{\psi}-$balancing]{\includegraphics[width=3.5cm]{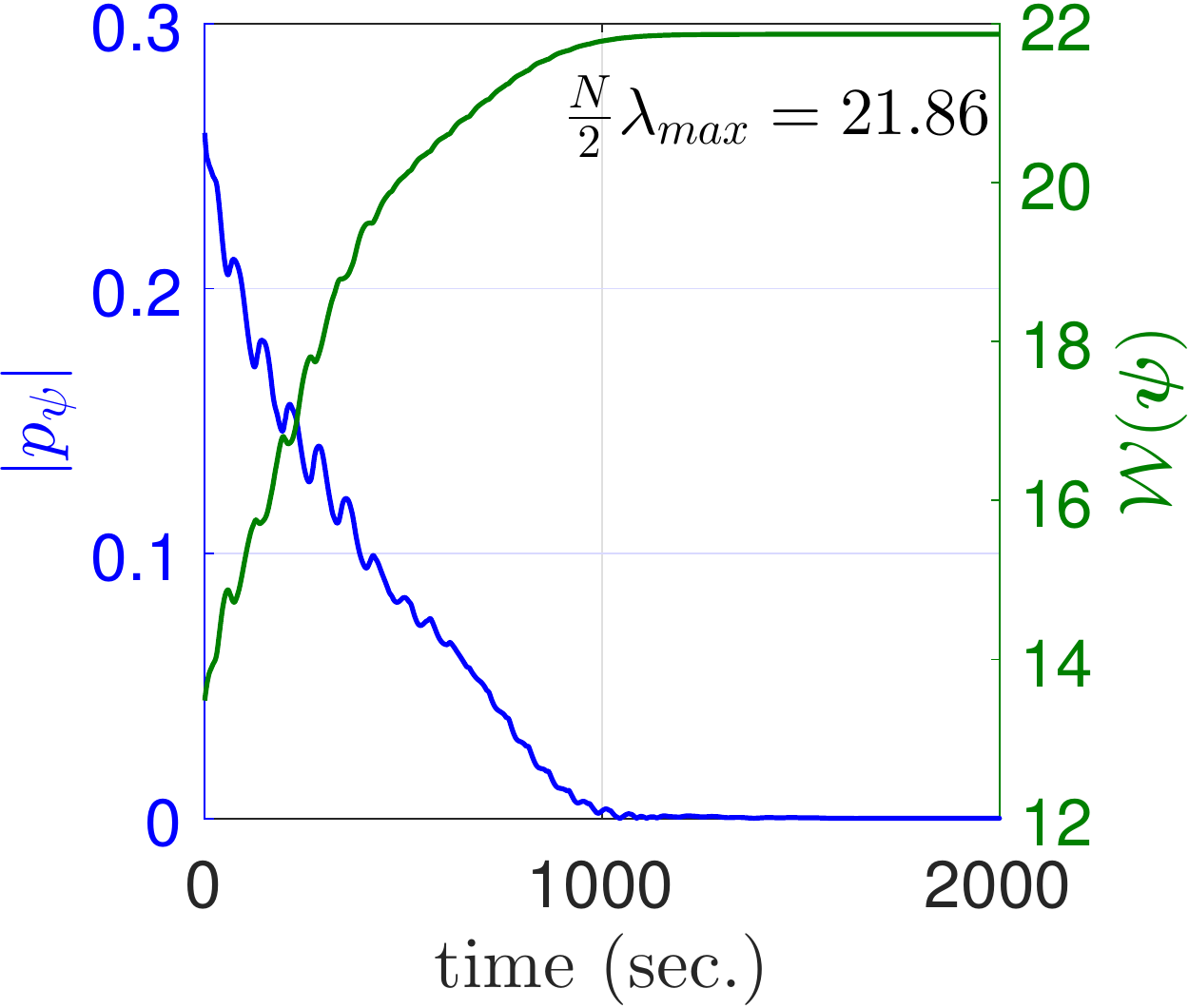}}\\
	\subfigure[$\mathcal{H}(\pmb{\psi})-$synchronization]{\includegraphics[width=3.5cm]{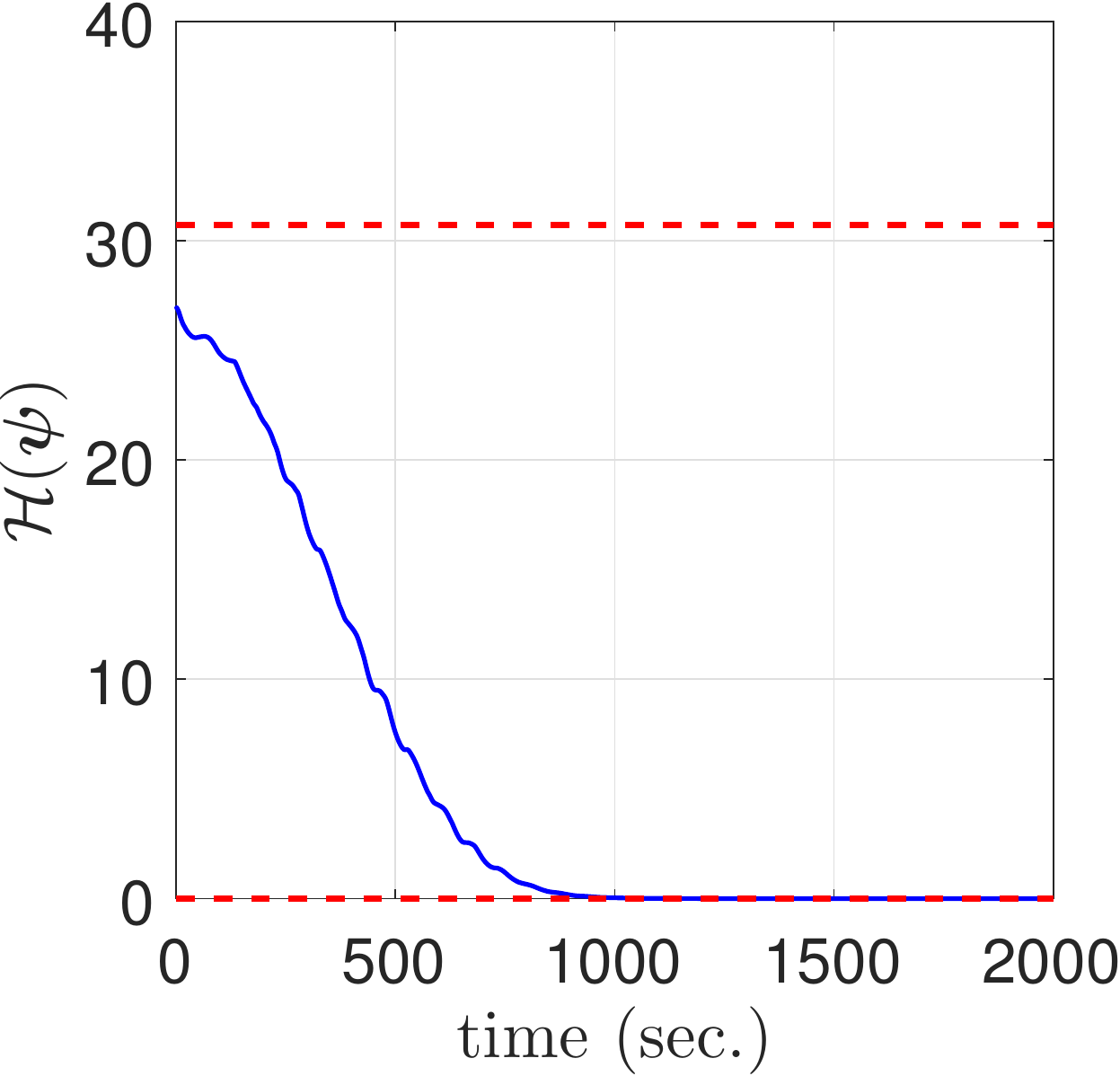}}\hspace{0.5cm}
	\subfigure[$\mathcal{H}(\pmb{\psi})-$balancing]{\includegraphics[width=3.5cm]{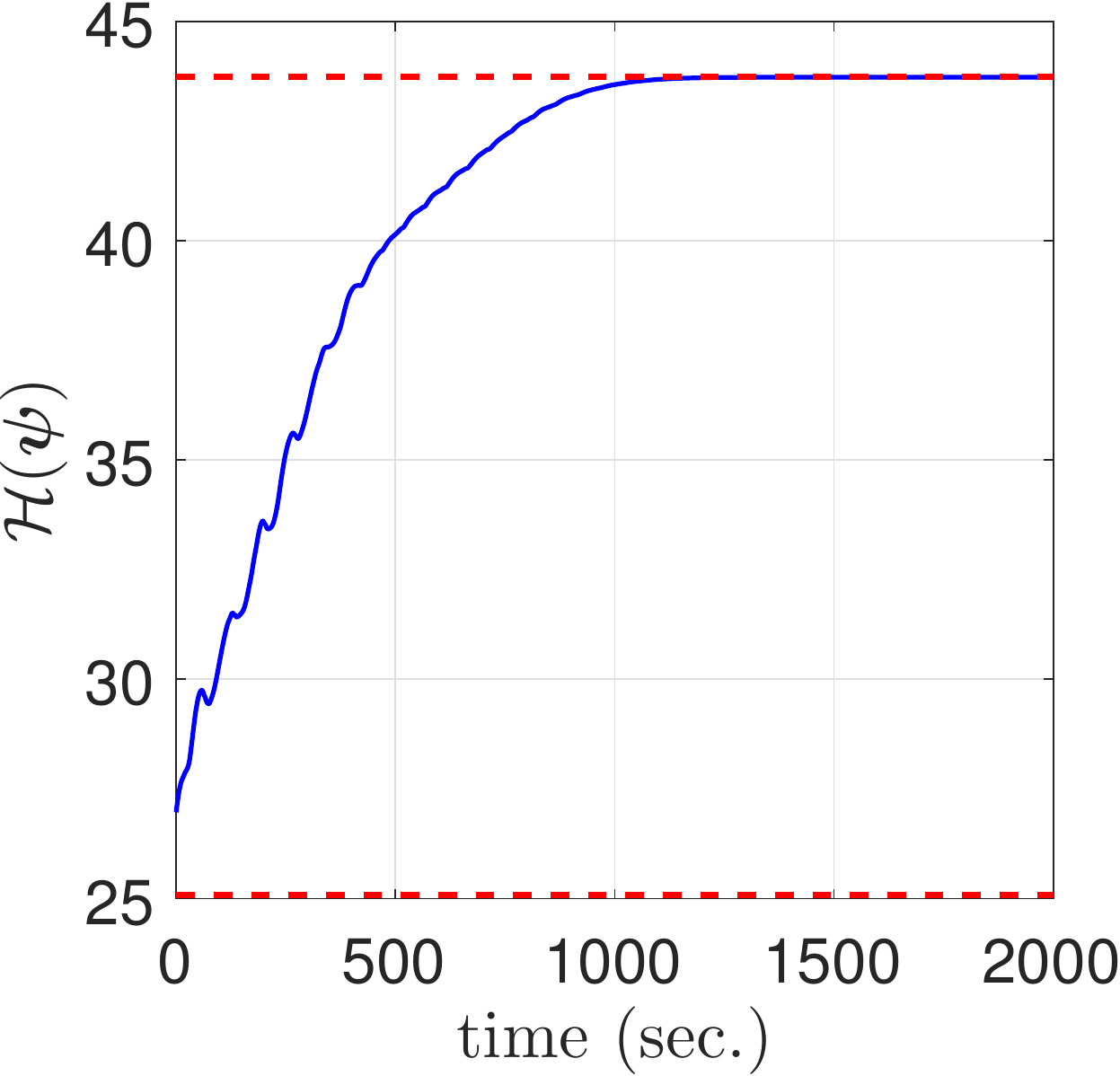}}
	\caption{Curve-phase characteristics for synchronization and balancing.}
	\label{phase_plots}
\end{figure}

\begin{figure}
	\centering
	\subfigure[$\theta_k-$synchronization]{\includegraphics[width=3.5cm]{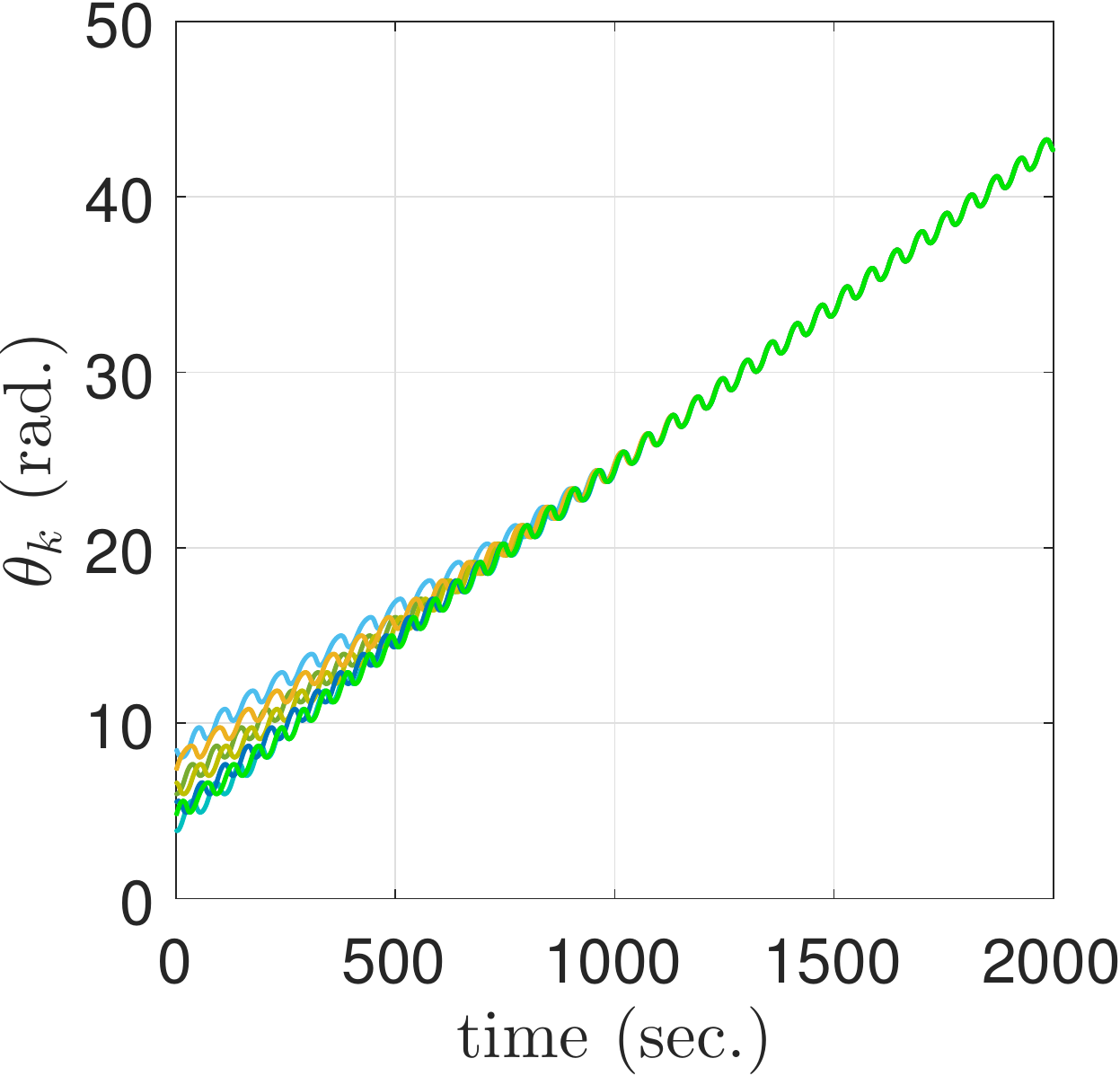}}\hspace{0.5cm}
	\subfigure[$\theta_k-$balancing]{\includegraphics[width=3.5cm]{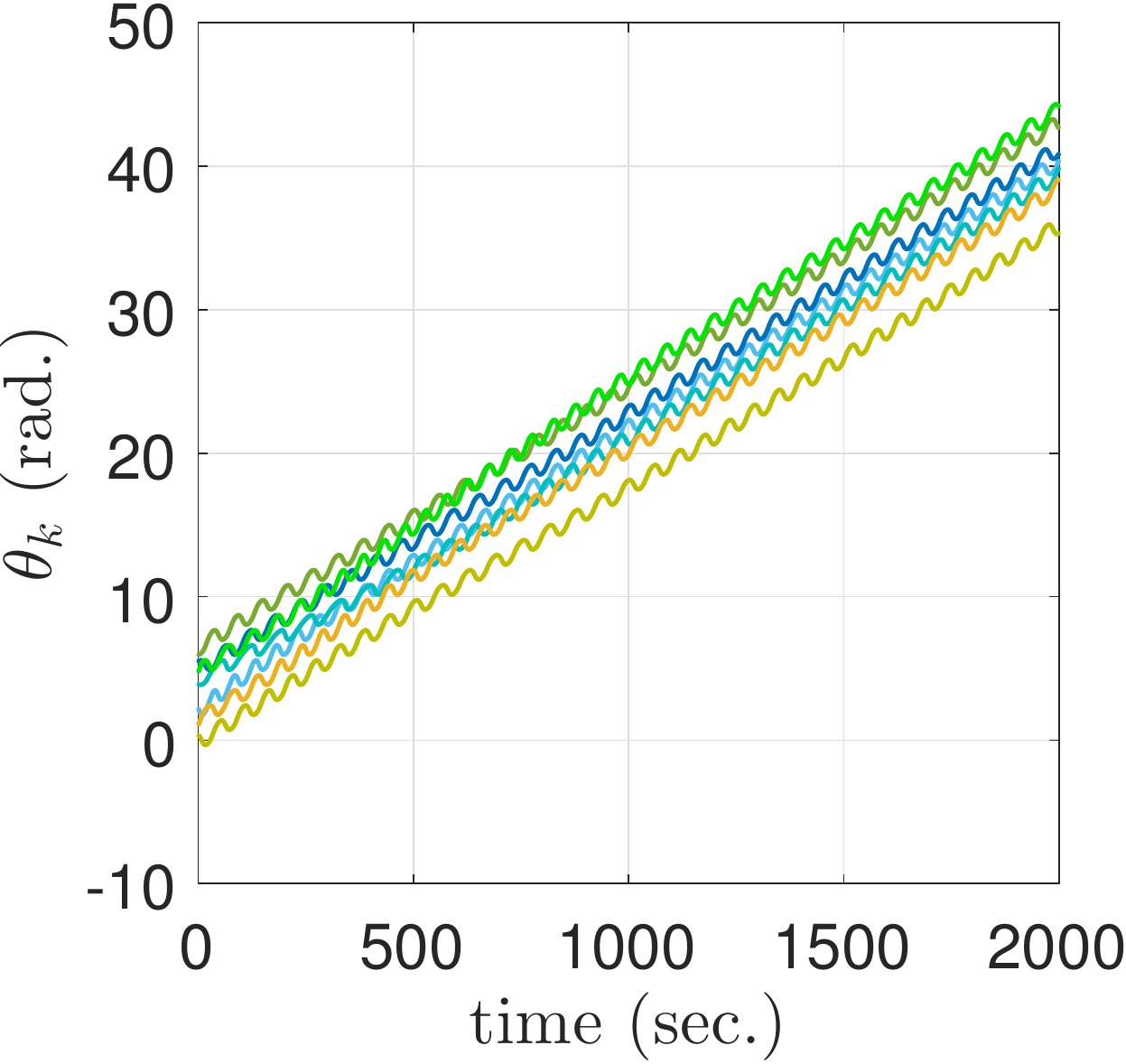}}\\
	\subfigure[$\psi_k-$synchronization]{\includegraphics[width=3.5cm]{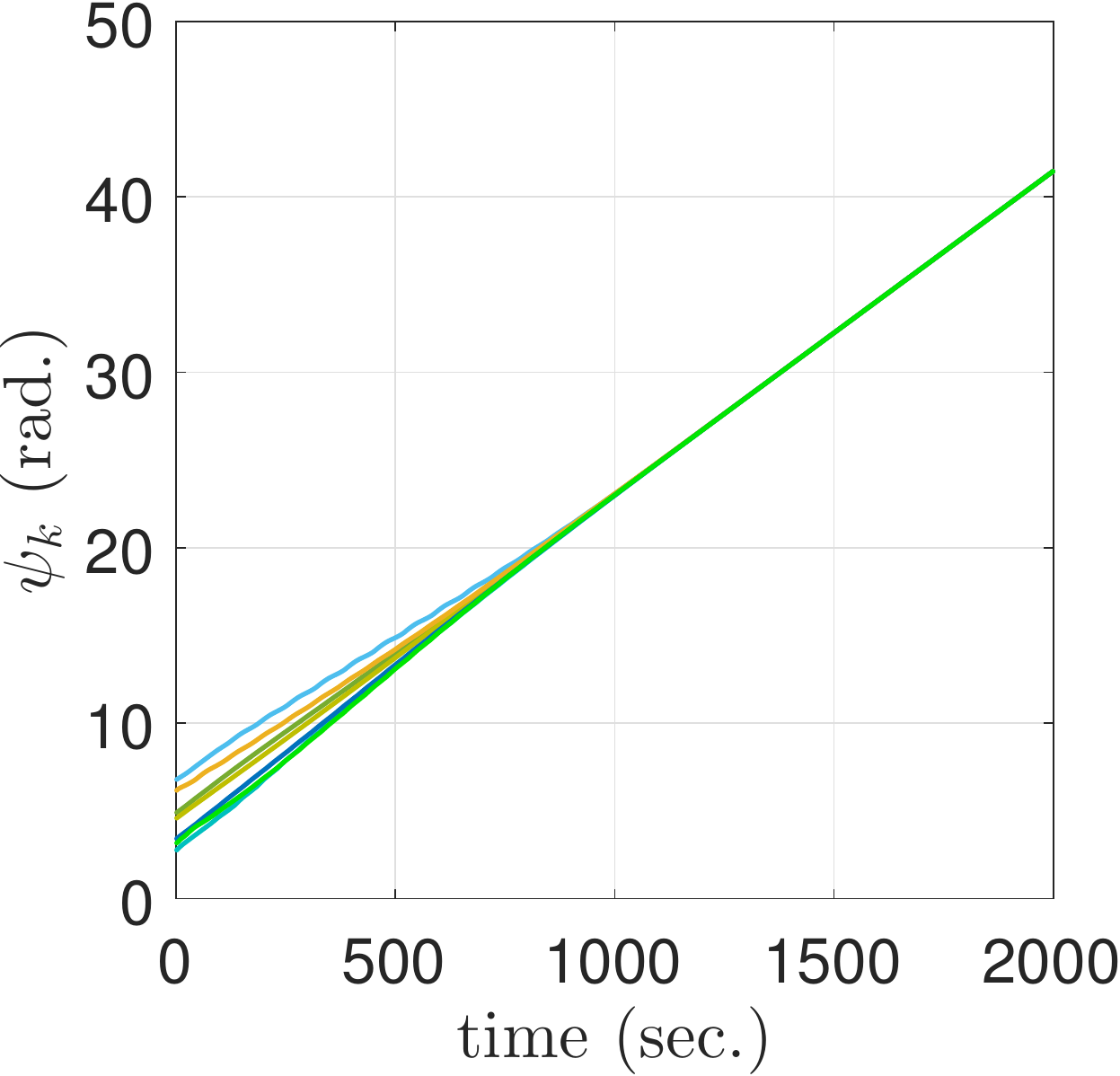}}\hspace{0.5cm}
	\subfigure[$\psi_k-$balancing]{\includegraphics[width=3.5cm]{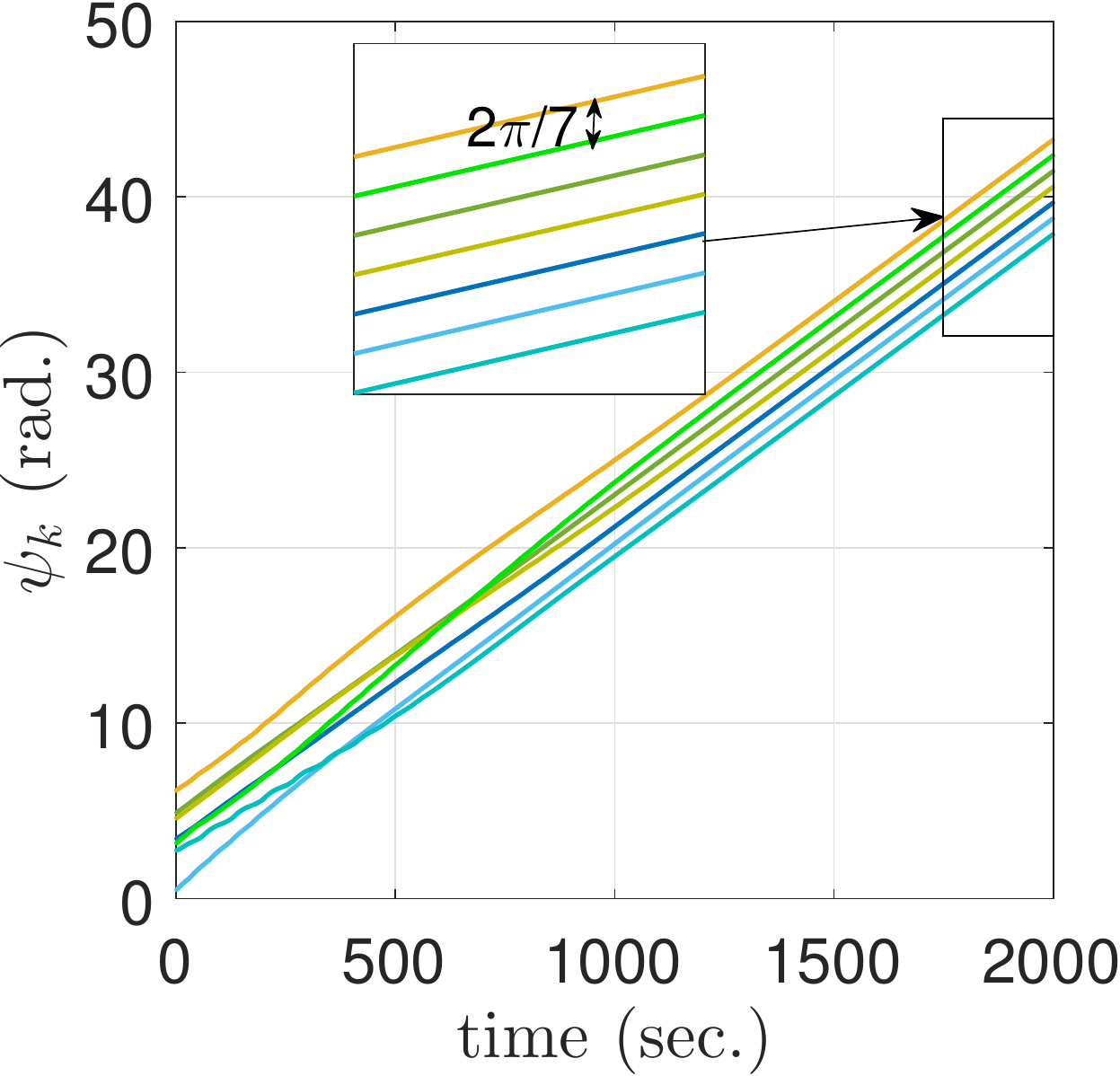}}\label{fig_psi_spaced}
	\caption{Agents' heading and curve-phases for synchronization and balancing.}
	\label{theta_psi_plots}
\end{figure}

\item Control inputs $u_k$ in \eqref{saturated_control} are depicted in Fig.~\ref{ctrl_plots} for all $k = 1, \ldots, 7$, where we assumed that the saturation limit is $u_{\text{max}} =  0.0786$ such that $u_{\text{max}} \geq \max_{\phi}|\kappa(\phi)| = 0.0776$ in \eqref{saturated_control}. Clearly, $|u_k| \leq u_{\text{max}}$ for all $k = 1, \ldots, 7$ in both curve-phase synchronization and balancing. An important observation in Fig.~\ref{ctrl_plots} is that the control inputs $u_k$ are also synchronized and phase-shifted in time for synchronization and balancing, respectively. Further, we also observe that $\zeta_k \rightarrow 0$ (and $u_k = \kappa(\phi)$), $\forall~k = 1, \ldots, 7$ in steady state, when the agents converge to the desired curve in synchronized and balanced phase patterns.

\item Fig.~\ref{phase_plots} sketches the magnitude $|p_{\psi}|$ of the average curve-phase momentum $p_{\psi}$, the curve-phase potential $\mathcal{W}(\pmb{\psi})$, and the quantity $\mathcal{H}(\pmb{\psi}) = \sum_{\{j,k\} \in \mathcal{E}}|{\rm e}^{i\psi_j} - {\rm e}^{i\psi_k}|^2$ for both curve-phase synchronization and balancing. It can be seen that $|p_{\psi}| \to 1$ and $\mathcal{W}(\pmb{\psi}) \to 0$ in case of synchronization. For balancing, $|p_{\psi}| \to 0$ and $\mathcal{W}(\pmb{\psi}) \to ({N}/{2})\lambda_{\text{max}}(\mathcal{L}) = 21.86 < 2|\mathcal{E}| = 26$, as discussed in Lemma \ref{lem_critical points of W} and Theorem~\ref{bounds_balancing}. Moreover, $\mathcal{H}(\pmb{\psi}) \in [0,30.7]$ for synchronization, and $\mathcal{H}(\pmb{\psi}) \in [25.1,43.7]$ for balancing. These bounds are calculated using Theorems \ref{bounds_synchronization} and \ref{bounds_balancing}, and are verified in Fig.~\ref{phase_plots}. 

\item In Fig.~\ref{theta_psi_plots}, we observe that the agent headings $\theta_k$ and curve phases $\psi_k$ converge in the case of synchronization, as the agents converge to the desired curve. For the balanced phase pattern of $N = 7$ agents, the curve-phases are spaced apart by $\frac{2\pi}{7}$ radians in the steady state (see Fig. \ref{theta_psi_plots}(d)).

\item We have numerically calculated the parameters and areas of the curves and boundaries. It is observed that $\Gamma_{\mathcal{C}} = 340.82$~(m), and $\Gamma_{\partial\mathcal{B}_{\delta}^-} = \Gamma_C - 2\pi\delta = 265.43~\text{(m)};~ \Gamma_{\partial\mathcal{B}_{\delta}^+} = \Gamma_C + 2\pi\delta = 416.21~\text{(m)}$, satisfying Theorem~\ref{thm_perimeter_area}. The area enclosed by the desired curve, and the inner and outer boundaries are also calculated numerically and agree with Theorem~\ref{thm_perimeter_area}, $\mathcal{A}_{\mathcal{C}} = 7893.3$~(m$^2$), $\mathcal{A}_{\partial\mathcal{B}_{\delta}^-} = \mathcal{A}_{\mathcal{C}} - \delta\Gamma_{\mathcal{C}} + \pi\delta^2 \approx 4255.7$~(m$^2$), and $\mathcal{A}_{\partial\mathcal{B}_{\delta}^+} = \mathcal{A}_{\mathcal{C}} + \delta\Gamma_{\mathcal{C}} + \pi\delta^2 \approx 12435.4$~(m$^2$). It is straightforward to check that $\Gamma_{\partial\mathcal{B}_{\delta}} = \Gamma_{\partial\mathcal{B}_{\delta}^+} + \Gamma_{\partial\mathcal{B}_{\delta}^-} = 2\Gamma_{\mathcal{C}} = 681.64~\text{(m)}$, and $\mathcal{A}_{\partial\mathcal{B}_{\delta}} = \mathcal{A}_{\partial\mathcal{B}_{\delta}^+} - \mathcal{A}_{\partial\mathcal{B}_{\delta}^-} = 2\delta\Gamma_{\mathcal{C}} \approx 8179.7$~(m$^2$). Furthermore, the inequalities $\Gamma^2_{\mathcal{C}} = 116158.6 > 99189.5 = 4\pi\mathcal{A}_{\mathcal{C}}$, and $\Gamma_{\mathcal{C}} = 340.82 > 80.87 = 2\pi/\kappa_{\text{max}}$, as mentioned in Remark~4, are also verified.
\end{itemize}

\section{Conclusion and Further Remarks}\label{section7}
Formation patterns of multi-agent systems in curve-phase synchronization and balancing around a desired simple closed curve, while considering two practical aspects$-$bounded trajectories and saturated control, were investigated in this paper. The concept of logarithmic BLF was used to derive the control laws. Using tools from Lyapunov stability theory and LaSalle's invariance principle, it was shown that the proposed controllers asymptotically stabilize the desired formation patterns around the desired simple closed polar curve, while the agents' trajectories remain bounded and the turn-rates obey the saturation limits. The analytical expressions for boundary, perimeter, and area of the trajectory-constraining set, were obtained under a mild assumption on the safe distance from the desired curve. Bounds on several signals of interest were derived and shown to be a function of initial conditions, control gains, and interaction topology among agents. Extensive MATLAB simulations were provided to illustrate the theoretical results. 

The issue of collision avoidance among agents is not addressed in this paper. In this work, the control input is realized through turn-rates of the vehicles. However, one will require a higher level of control efforts to tackle collision avoidance \cite{panagou2013multi,panagou2015distributed,han2019robust}. The incorporation of practical aspects like communication time-delays, directed and dynamically changing interaction topology, external disturbances, etc., constitute an interesting and indeed a challenging future scope of the work, due to nonlinear nature of the control laws.

\section*{Acknowledgments}  
The authors would like to gratefully acknowledge Prof. Debasish Ghose for his helpful comments and suggestions.


\bibliographystyle{ieeetran}
\bibliography{References}

\end{document}